\documentclass[cleveref, numberwithinsect, a4paper]{lipics-v2021}

\usepackage[amsalign, lipics, nonatbib]{prelude}
\usepackage[noadjust]{cite}

\nolinenumbers

\input{figures/tikz-macros}

\newcommand{\nosubfigure}{\captionsetup{aboveskip=0.25em}}

\newcommand{\srpt}{\mathrm{SRPT}}
\newcommand{\psjf}{\mathrm{PSJF}}
\newcommand{\srpte}{\mathrm{SRPT\mathhyphen E}}
\newcommand{\psjfe}{\mathrm{PSJF\mathhyphen E}}
\newcommand{\srptb}{\mathrm{SRPT\mathhyphen B}}
\newcommand{\srptse}{\mathrm{SRPT\mathhyphen SE}}
\newcommand{\srptcue}{\mathrm{SRPT\mathhyphen CUE}}

\newcommand{\generic}{\pi}
\newcommand{\genericOther}{{\pi'}}

\newcommand{\size}{\mathsf{size}}
\newcommand{\remsize}{\mathsf{remsize}}
\newcommand{\sizee}{\mathsf{size\mathhyphen e}}
\newcommand{\remsizee}{\mathsf{remsize \mathhyphen e}}
\newcommand{\func}{\mathsf{func}}

\newcommand{\rank}[1]{\mathsf{rank}_{#1}}
\newcommand{\worst}[1]{\mathsf{rank}^{\mathsf{worst}}_{#1}}

\newcommand{\policies}{\{\srpte,\allowbreak \psjfe,\allowbreak \srptb,\allowbreak \srptse\}}
\newcommand{\policiesWithTrue}{\{\srpt,\allowbreak \psjf,\allowbreak \srpte,\allowbreak \psjfe,\allowbreak \srptb,\allowbreak \srptse\}}

\newcommand{\Twait}{T^{\mathsf{wait}}}
\newcommand{\Tres}{T^{\mathsf{res}}}

\newcommand{\by}[2][0mu]{&&\mkern#1\textenv{\footnotesize [#2]}}
\newcommand{\nestint}{\mkern-6mu\int}

\newlength{\fixmidfigure}
\newcommand{\premidfigure}{\setlength{\fixmidfigure}{\lastskip}\addvspace{-\lastskip}}
\newcommand{\postmidfigure}{\addvspace{\fixmidfigure}}

\makeatletter
\let\@oldcite\cite
\renewcommand{\cite}[2][]{%
  {\IfSubStr{#2}{,}{}{\ifempty{#1}{\unskip~}{}}%
    \ifempty{#1}{\@oldcite{#2}}{\@oldcite[#1]{#2}}}}

\let\@oldendtikzpicture\endtikzpicture
\renewcommand{\endtikzpicture}{\@oldendtikzpicture\ignorespacesafterend}
\makeatother

\newcommand{\lookupauthors}[1]{%
    \IfEqCase{#1}{%
        {akbari-moghaddam_seh_2021}{Akbari-Moghaddam and Down}%
        {azar_flow_2021}{Azar et\penalty50\ al.}%
        {purohit_improving_2018}{Purohit et\penalty50\ al.}%
        {mitzenmacher_scheduling_2020}{Mitzenmacher}%
        {scully_soap_2018}{Scully et\penalty50\ al.}%
        {scully_soap_2018-1}{Scully and Harchol-Balter}%
        {scully_gittins_2020}{Scully et\penalty50\ al.}%
        {wierman_nearly_2005}{Wierman et\penalty50\ al.}%
        {scully_gittins_2021}{Scully and Harchol-Balter}%
        {dellamico_psbs_2016}{Dell'Amico et\penalty50\ al.}%
        {mitzenmacher_supermarket_2020}{Dell'Amico and Mitzenmacher}%
    }[\GenericError{}{unknown key in lookupauthors}{}{add it to the big IfEqCase}\textbf{??}]}
\newcommand{\citet}[2][]{\lookupauthors{#2}\ifempty{#1}{~\cite{#2}}{ \cite[#1]{#2}}}
\newcommand{\Citet}{\citet}

\title{Uniform Bounds for Scheduling with Job Size Estimates}

\author{Ziv Scully}{Carnegie Mellon University, Computer Science Department, Pittsburgh, PA, USA}{zscully@cs.cmu.edu}{0000-0002-8547-1068}{}

\author{Isaac Grosof}{Carnegie Mellon University, Computer Science Department, Pittsburgh, PA, USA}{igrosof@cs.cmu.edu}{0000-0001-6205-8652}{}

\author{Michael Mitzenmacher}{Harvard University, School of Engineering and Applied Sciences, Cambridge, MA, USA}{michaelm@eecs.harvard.edu}{0000-0001-5430-5457}{Supported in part by NSF grants CCF-2101140, CNS-2107078, and DMS-2023528, and by a gift to the Center for Research on Computation and Society at Harvard University.}

\funding{\textit{Ziv Scully and Isaac Grosof}: Supported by NSF grants CMMI-1938909 and CSR-1763701.}

\authorrunning{Z.~Scully, I.~Grosof, and M.~Mitzenmacher}
\Copyright{Ziv Scully, Isaac Grosof, and Michael Mitzenmacher}

\relatedversiondetails{arXiv Mirror}{https://arxiv.org/abs/2110.00633}



\EventEditors{Mark Braverman}
\EventNoEds{1}
\EventLongTitle{13th Innovations in Theoretical Computer Science Conference (ITCS 2022)}
\EventShortTitle{ITCS 2022}
\EventAcronym{ITCS}
\EventYear{2022}
\EventDate{January 31--February 3, 2022}
\EventLocation{Berkeley, CA, USA}
\EventLogo{}
\SeriesVolume{215}
\ArticleNo{114}

\begin{document}

\maketitle

\begin{abstract}
We consider the problem of scheduling to minimize mean response time in M/G/1 queues where only \emph{estimated} job sizes (processing times) are known to the scheduler, where a job of true size~$s$ has estimated size in the interval $[\beta s, \alpha s]$ for some $\alpha \geq \beta > 0$. We evaluate each scheduling policy by its \emph{approximation ratio}, which we define to be the ratio between its mean response time and that of Shortest Remaining Processing Time (SRPT), the optimal policy when true sizes are known. Our question: is there a scheduling policy that (a) has approximation ratio near 1 when $\alpha$ and $\beta$ are near~1, (b) has approximation ratio bounded by some function of $\alpha$ and~$\beta$ even when they are far from~1, and (c) can be implemented without knowledge of $\alpha$ and~$\beta$?

We first show that naively running SRPT using estimated sizes in place of true sizes is \emph{not} such a policy: its approximation ratio can be arbitrarily large for any fixed $\beta < 1$. We then provide a simple variant of SRPT for estimated sizes that satisfies criteria (a), (b), and~(c). In particular, we prove its approximation ratio approaches 1 uniformly as $\alpha$ and $\beta$ approach 1. This is the first result showing this type of convergence for M/G/1 scheduling.

We also study the Preemptive Shortest Job First (PSJF) policy, a cousin of SRPT. We show that, unlike SRPT, naively running PSJF using estimated sizes in place of true sizes satisfies criteria (b) and~(c), as well as a weaker version of~(a).
\end{abstract}

\section{Introduction}

Minimizing mean response time of jobs in a preemptive single-server queue
is a fundamental scheduling problem.
If the scheduler knows each job's size,\footnote{A job's size is its processing time.}
then the optimal policy is \emph{Shortest Remaining Processing Time} (SRPT),
which always serves the job of least remaining size:
least size minus time served so far.
However, in practical queueing systems,
it is rare that the scheduler knows each job's exact size, which is required for SRPT.
Instead, it may be that the scheduler has only an \emph{estimated size} for each job.

In settings where the scheduler knows only job size estimates,
rather than true sizes,
how should one schedule to minimize mean response time?
We study this question in a stochastic online setting, namely an M/G/1 queue,\footnote{%
The M/G/1 is a queueing model with Poisson arrivals and i.i.d. job sizes.
We define the M/G/1 formally in \cref{sec:model}.}
in which jobs arrive randomly over time.
We use $T$ to denote the distribution of response time,
and we seek policies that achieve strong guarantees on $\E{T}$, the mean response time.
We focus our attention on simple policies that do not depend on
detailed knowledge of the distributions of true or estimated job sizes,
as such knowledge may not be available in practice.
While some simple heuristics have been proposed in prior work,
no performance guarantees have been proven for these policies.
We evaluate each scheduling policy by its \emph{approximation ratio}, which we define to be the ratio between its mean response time and that of Shortest Remaining Processing Time (SRPT), the optimal policy when true sizes are known.

In the context of online algorithms with predictions \cite{DBLP:journals/jacm/LykourisV21,DBLP:books/cu/20/MitzenmacherV20},
policy design has two key goals: ``consistency'',
which requires near-optimal performance under low error,
and ``robustness'', which requires bounded approximation ratio
under arbitrary error.
Unfortunately, as we explain in \cref{sec:gittins_bound},
this type of robustness is provably unachievable in this context.
Instead, we focus on a more appropriate guarantee,
which we call ``graceful degredation'',
which requires that the performance degrades smoothly and slowly
as error increases.

We focus on the setting of multiplicative errors:
We assume that a job of true size~$s$ has estimated size
in the interval $[\beta s, \alpha s]$ for some $\alpha \geq \beta > 0$.  We refer to this assumption by saying the jobs have $(\beta,\alpha)$-bounded estimates.  (Here $\beta$ stands for ``below'' and $\alpha$ for ``above.'')
In this context, a policy~$\generic$ is consistent if $\E{T_\generic} \to \E{T_\srpt}$ in the limit as $\alpha, \beta \to 1$.
Graceful degradation requires that for some constant~$C$,
we have $\E{T_\generic} \le C\frac{\alpha}{\beta} \E{T_\srpt}$
for all $\alpha \geq \beta > 0$.

In trying to achieve consistency and graceful degradation,
a first policy one might consider is the \emph{SRPT with Estimates} (SRPT-E) policy,
which always serves the job of least estimated size minus time served so far.
Unfortunately, as we prove in \cref{thm:srpt-e_result},
SRPT-E can have infinite approximation ratio if $\beta < 1$,
so it has neither consistency nor graceful degradation.

In this work, we present the first policy which provably has both consistency and graceful degradation:
the \emph{SRPT with Bounce} (SRPT-B) policy,
defined in \cref{sec:model}.
Specifically, we show that SRPT-B's approximation ratio is at most $3.5\alpha/\beta$
and approaches~$1$ uniformly as $\alpha, \beta \to 1$.
We also study the policy \emph{Preemptive Shortest Job First with Estimates} (PSJF-E),
for which we prove even better graceful degradation guarantees,
and consistency guarantees relative to the perfect-information PSJF policy:
we show $\E{T_\psjfe} \leq \frac{\alpha}{\beta} \E{T_\psjf}$,
which turns out to imply that PSFJ-E's approximation ratio is at most~$1.5\alpha/\beta$.

See \cref{sec:model} for the details of our queueing model
and the definitions of the scheduling policies we consider,
\cref{sec:results} for a full statement of our main results,
and \cref{sec:proof_overview} for a high-level overview of our proof techniques.

\subsection{Related Work}

Policies for ordering jobs according to their service time have been studied extensively in single queues.  The text \cite{harchol-balter_performance_2013} provides an excellent introduction to the analysis of standard policies, such as Shortest Job First (SJF), PSJF, and SRPT, in the single-queue setting.

Settings where estimates or predictions of service times, such as one might obtain from machine learning algorithms, have been much less studied.  The work closest to ours is that of Wierman and Nuyens \cite{wierman_scheduling_2008}.
They study policies that they dub $\epsilon$-SMART policies.
Such policies include variations of SRPT and PSJF
with inexact job sizes,
and they bound the performance
of such policies
based on how inexact the estimates can be \cite{wierman_scheduling_2008}.
However, their results only apply to two simpler types of error:
addditive error, where the estimate of a job of size $s$
is within $[s-\sigma, s+\sigma]$;
and speedup error, where the estimate is updated as the job runs,
and the estimate of a job of remaining size $r$ is within $[(1-\sigma)r, (1+\sigma)r]$.
In contrast, we primarily focus on the more realistic setting of
multiplicative non-updating error.
Their results on mean response time demonstrate only graceful degradation,
are restricted to the setting of speedup error,
and require additional assumptions on the job size distribution.
While it is not our primary focus, we also discuss
using the techniques in this paper
to achieve consistency and graceful degradation in the speedup-error setting in \cref{sec:srpt-se, sec:continuous-estimate}.
Such results do not require additional assumptions on the job size distribution.

\Citet{dellamico_psbs_2016} empirically study scheduling policies for queueing systems with estimated sizes; \Citet{mitzenmacher_supermarket_2020} also perform an empirical study of scheduling policies with estimated sizes, but in the context of multiple queues using the power of two choices.  Mitzenmacher provides formulas for the mean response time for M/G/1 queues under scheduling policies where
service times are predicted rather than known
exactly according to a stochastic model, including for our SRPT-E and PSJF-E policies \cite{mitzenmacher_scheduling_2020}.\footnote{In \cite{mitzenmacher_scheduling_2020} the schemes using predictions are referred to as SPRPT (shortest predicted remaining processing time), and PSPJF;  there does not seem to be a consistent nomenclature for policies with predictions/estimates, and we hope our labeling is more readily understood.}
In later work Mitzenmacher studies similar models where only a single bit of prediction-based advice is given, and also studies single-bit advice in the mean-field setting under the power of two choices \cite{doi:10.1137/1.9781611976830.1}.  Several of these works note that SRPT-E performance can degrade in situations where job sizes have high variance, and that PSJF-E can have better performance in these cases \cite{dellamico_psbs_2016,mitzenmacher_scheduling_2020,mitzenmacher_supermarket_2020}.

\Citet{akbari-moghaddam_seh_2021} introduce two new heuristic policies for scheduling given size estimates. They demonstrate empirically that the new heuristics overcome the shortcomings of SRPT-E. However, unlike PSJF-E, they do not give jobs a static priority, which often leads to better performance than PSJF-E. It is an interesting open question whether our techniques, which give bounds on PSJF-E's performance, could also bound the performance of the heuristics introduced by \citet{akbari-moghaddam_seh_2021}.

In the setting of scheduling with predictions for finite collections of jobs,
combinations of shortest predicted job first and round robin that yield good performance in terms of the competitive ratio were studied by \citet{purohit_improving_2018}.
For the online scheduling problem of weighted mean response time on a single machine
with a finite arrival sequence, \citet{azar_flow_2021} consider what we would refer to as $(1,\mu)$-bounded jobs where $\mu$ is known, and give algorithms that are competitive up to a logarithmic factor in the maximum ratios of the processing times, densities, and weights.
For unweighted mean response time,
they prove a variant of graceful degradation: if job size estimates have a multiplicative error of at most $\mu$, they prove an $O(\mu^2)$
competitive ratio bound. In contrast, our graceful degradation bounds are linear in $\mu = \frac{\alpha}{\beta}$, and do not depend on knowledge of $\alpha$ or $\beta$.\footnote{Following submission of this paper, the authors of \cite{azar_flow_2021} released a paper with improved results \cite{azar2021distortionoblivious}.  Specifically, they provide a new algorithm and prove an $O(\mu \log^2 \mu)$
competitive ratio bound for that algorithm, which further does not depend on knowledge of $\alpha$ or $\beta$.}

\section{Model and Preliminaries}
\label{sec:model}

We now define our model and provide basic notation and definitions.
Further definitions and background results appear in \cref{sec:background}.

We consider a stochastic scheduling setting called the M/G/1 queue with job size estimates.
The ``M'' in ``M/G/1'' refers to the assumption that jobs arrive according to a Poisson process (that is, with exponentially distributed interarrival times) with rate $\lambda$.
The ``G'' in ``M/G/1'' refers to the assumption that job sizes are sampled i.i.d. from a general distribution.
We additionally assume that each job $j$
has a size $s_j$ and an estimated size $z_j$, where the pair $(s_j, z_j)$
is sampled i.i.d. from some joint distribution $(S, Z)$.
We assume $(S, Z)$ is a continuous distribution
with joint density function $f_{S, Z}(s, z)$.
We discuss this continuity assumption further in \cref{sec:ties}.
We write $f_S(s)$ and $f_Z(z)$ for the marginal densities
of $S$ and~$Z$, respectively.
Regardless of scheduling policy,
the fraction of time the server is busy, also known as the load, is fixed.
We denote load by $\rho$ and note that $\rho = \lambda \E{S}$.
We assume that $\rho < 1$ to ensure that the server completes jobs faster than they arrive in the long run.
The ``1'' in ``M/G/1'' refers to there being a single server.

We focus on the setting of multiplicatively-bounded size estimates:
\begin{definition}
    Size estimates are $(\beta, \alpha)$-bounded
    for some constants $0 < \beta \le \alpha$
    if for all jobs~$j$,
    \begin{align*}
        z_j \in [\beta s_j, \alpha s_j].
    \end{align*}
\end{definition}

Mnemonically, $\beta$ is the bound below, and $\alpha$ is the bound above.

\begin{definition}
    \label{def:state}
    The \emph{state} of job~$j$ is the triple $x_j = (s_j, z_j, a_j)$
    consisting of
    \begin{itemize}
        \item its \emph{(true) size}~$s_j\esub$,
        which is the amount of time it must be served to complete;
        \item its \emph{estimated size}~$z_j\esub$,
        which is revealed when the job arrives
        and is guaranteed to be in the interval $[\beta s_j, \alpha s_j]$; and
        \item its \emph{age}~$a_j\esub$, the amount of service the job has received so far.
    \end{itemize}
    Job~$j$ completes once $a_j = s_j\esub$.
\end{definition}

\begin{figure}[t]
    \begin{subfigure}[t]{0.5\linewidth}
        \begin{tikzpicture}[figure]
    \yguide[\hphantom{estimated size~$z$}\llap{true size~$s$}]{0}{5.5}
    \xguide[$s$]{5.5}{0}
    \axes{10}{6}{$0$}{age}{$0$}{rank}
    \draw[srpt] (0, 5.5) -- (5.5, 0);
\end{tikzpicture}
        \caption{Shortest Remaining Processing Time (SRPT)}
        \label{fig:rank:srpt}
    \end{subfigure}\hfill\begin{subfigure}[t]{0.5\linewidth}
        \begin{tikzpicture}[figure]
    \yguide[\hphantom{estimated size~$z$}\llap{true size~$s$}]{0}{5.5}
    \axes{10}{6}{$0$}{age}{$0$}{rank}
    \draw[psjf] (0, 5.5) -- (10, 5.5);
\end{tikzpicture}
        \caption{Preemptive Shortest Job First (PSJF)}
        \label{fig:rank:psjf}
    \end{subfigure}\\[1em]\begin{subfigure}[t]{0.5\linewidth}
        \begin{tikzpicture}[figure]
    \yguide[estimated size~$z$]{0}{4}
    \xguide[$z$]{4}{0}
    \axes{10}{6}{$0$}{age}{$0$}{rank}
    \draw[srpt-e] (0, 4) -- (5.5, -1.5);
\end{tikzpicture}
        \caption{SRPT with Estimates (SRPT-E)}
        \label{fig:rank:srpt-e}
    \end{subfigure}\hfill\begin{subfigure}[t]{0.5\linewidth}
        \begin{tikzpicture}[figure]
    \yguide[estimated size~$z$]{0}{4}
    \axes{10}{6}{$0$}{age}{$0$}{rank}
    \draw[psjf-e] (0, 4) -- (10, 4);
\end{tikzpicture}
        \caption{PSJF with Estimates (PSJF-E)}
        \label{fig:rank:psjf-e}
    \end{subfigure}\\[1em]\begin{subfigure}[t]{0.5\linewidth}
        \begin{tikzpicture}[figure]
    \yguide[estimated size~$z$]{8}{4}
    \xguide[$z$]{4}{0}
    \axes{10}{6}{$0$}{age}{$0$}{rank}
    \draw[srpt-b] (0, 4) -- (4, 0) -- (8, 4) -- (10, 4);
\end{tikzpicture}
        \caption{SRPT with Bounce (SRPT-B)}
        \label{fig:rank:srpt-b}
    \end{subfigure}\hfill\begin{subfigure}[t]{0.5\linewidth}
        \begin{tikzpicture}[figure]
    \yguide[estimated size~$z$]{0}{4}
    \xguide[true size~$s$]{7}{0}
    \axes{10}{6}{$0$}{age}{$0$}{rank}
    \draw[srpt-se] (0, 4) -- (7, 0);
\end{tikzpicture}
        \caption{SRPT with Scaling Estimates (SRPT-SE)}
        \label{fig:rank:srpt-se}
    \end{subfigure}
    \caption{Rank Functions of Size-Estimate-Based Policies}
    \label{fig:rank}
\end{figure}
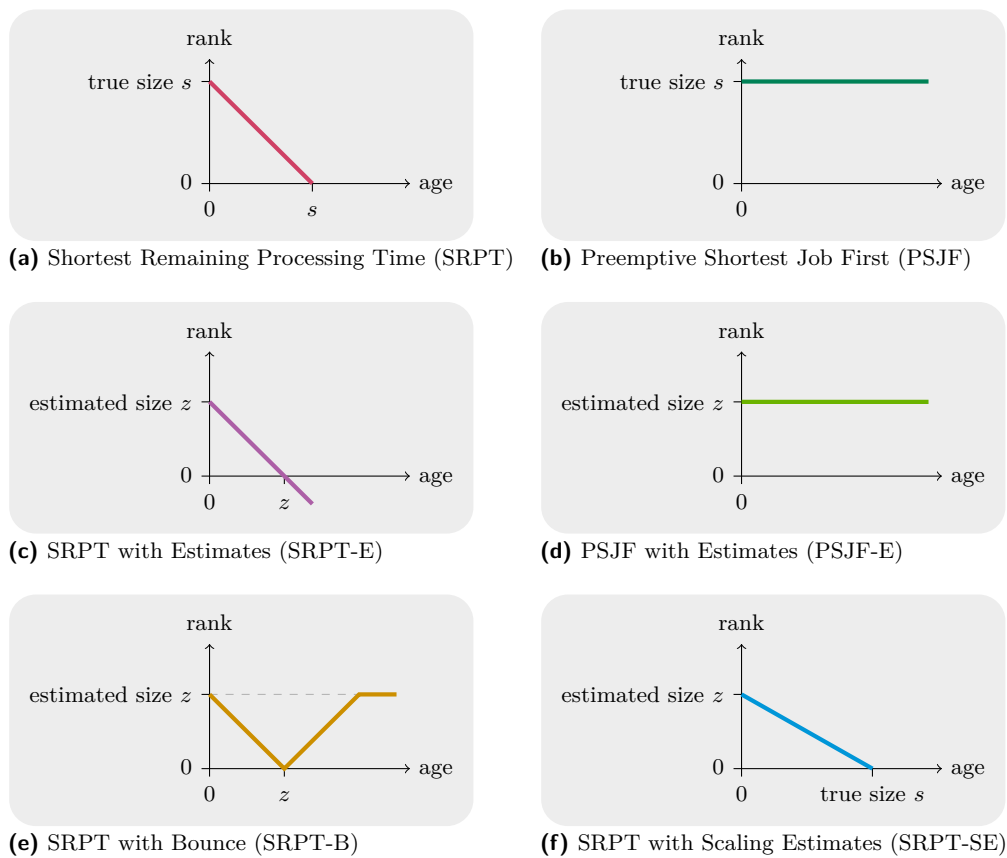

We now formally define the scheduling policies we consider.

\begin{definition}
    \label{def:policy}
    We consider six scheduling policies in this work.
    We define each policy~$\generic$ by a \emph{rank function},
    denoted $\rank{\generic}(x)$ or $\rank{\generic}(s, z, a)$
    assigning a \emph{rank}, or priority,
    to a job based on its state $x = (s, z, a)$.
    The scheduler always serves whichever job has the least rank.\footnote{%
    We tiebreak arbitrarily. Given our continuous job-size assumption and our specific policies, ties happen with probability zero. See \cref{sec:ties} for details.}
    The policies, which we illustrate in \cref{fig:rank}, are the following:
    {\setlength{\jot}{0em}\newcommand{\spacing}{\mkern30mu}
    \begin{alignat*}{3}
        & \text{\emph{Shortest Remaining Processing Time} (SRPT)} \spacing &
        & \rank{\srpt}(s, z, a) & &= s - a, \\
        & \text{\emph{Preemptive Shortest Job First} (PSJF)} \spacing &
        & \rank{\psjf}(s, z, a) & &= s, \\
        & \text{\emph{SRPT with Estimates} (SRPT-E)} \spacing &
        & \rank{\srpte}(s, z, a) & &= z - a, \\
        & \text{\emph{PSJF with Estimates} (PSJF-E)} \spacing &
        & \rank{\psjfe}(s, z, a) & &= z, \\
        & \text{\emph{SRPT with Bounce} (SRPT-B)} \spacing &
        & \rank{\srptb}(s, z, a) & &= \min\{|z - a|, z\}, \\
        & \text{\emph{SRPT with Scaling Estimates} (SRPT-SE)} \spacing &
        & \rank{\srptse}(s, z, a) & &= z/s \cdot (s - a).
    \end{alignat*}}
\end{definition}

For SRPT, shown in \cref{fig:rank:srpt} the rank of a job is the {\em remaining size} $s-a$, while for SRPT-E, shown in \cref{fig:rank:srpt-e}, the rank is
the {\em estimated remaining size} $z-a$, using the estimate in place of the true size.  (Note the rank for SRPT-E can be negative.)  Similarly, PSJF-E's rank function $z$ uses the estimate where PSJF uses the true size $s$ for its rank, as shown in \cref{fig:rank:psjf, fig:rank:psjf-e}.

For SRPT-B, shown in \cref{fig:rank:srpt-b}, the rank is given by $\min\{|z-a|,z\}$.
SRPT-B's rank is identical to SRPT-E's rank for ages $a \in [0, z]$,
but rises back to $z$ for larger $a$, yielding the bounce and thus the name SRPT with Bounce.

As a theoretical tool, we will also consider SRPT with Scaling Estimates, or SRPT-SE, shown in \cref{fig:rank:srpt-se}, for which the rank function is $z/s \cdot (s-a)$, a horizontally stretched version of SRPT-E.
Note that SRPT-SE is not implementable in our model,
as the scheduler does not have access to the true size~$s$.

\section{Description of Main Results}
\label{sec:results}

As discussed in the introduction,
our goal is to derive the first provably consistent and gracefully-degrading
policies in the size-estimate setting.
In the setting of $(\beta, \alpha)$-bounded size estimates,
consistency requires that in the $\beta, \alpha \to 1$ limit,
the policy achieves optimal mean response time,
matching that of SRPT, the optimal known-size policy.
``Graceful degradation'' requires that a policy's mean response
is bounded relative to that of SRPT and the $\alpha$ and $\beta$ values:
\begin{align*}
    \E{T_\generic} \le C \frac{\alpha}{\beta} \E{T_\srpt},
\end{align*}
for some constant~$C$.
Robustness, in the sense of achieving constant approximation ratio for arbitrary errors, is not possible in this setting, as we discuss in \cref{sec:gittins_bound}.

First, we show that consistency and graceful degradation are not straightforward to achieve.
Our first result shows that simply using SRPT-E (SRPT with Estimates),
a natural policy studied previously \cite{mitzenmacher_scheduling_2020},
yields mean response times that are not bounded within a constant factor of SRPT
in the worst case, even with $(\beta, \alpha)$-bounded size estimates.

\begin{restatable*:theorem}[Performance of SRPT-E]
    \label{thm:srpt-e_result}
    Consider the M/G/1 with $(\beta, \alpha)$-bounded size estimates.
    \begin{enumerate}[(a)]
    \item
        For any size distribution~$S$,
        there exists a joint distribution $(S, Z)$ of true and estimated sizes such that
        the mean response time of SRPT-E is bounded below by
        \begin{equation*}
            \E{T_\srpte} \geq \frac{\lambda(1 - \beta)^2}{2} \E{S^2},
        \end{equation*}
    \item
        The approximation ratio of SRPT-E may be arbitrarily large or infinite
        whenever $\beta < 1$.
    \end{enumerate}
\end{restatable*:theorem}

We consider a novel variation of SRPT, SRPT with Bounce (SRPT-B),
and prove it is consistent and gracefully-degrading, without knowledge of $\alpha$ and $\beta$.
This is the first proof of a policy satisfying these criteria.
Here, as $\alpha$ and $\beta$ approach 1, SRPT-B approaches the performance of SRPT, and it achieves a suitable finite approximation ratio for all $\alpha$ and $\beta$.
Formally, we prove the following:

\begin{restatable*:theorem}[Performance of SRPT-B]
    \label{thm:srpt-b_result}
    Consider the M/G/1 with $(\beta, \alpha)$-bounded size estimates.
    \begin{enumerate}[(a)]
    \item
        The mean response time of SRPT-B is bounded above by
        \begin{align*}
            \E{T_\srptb}
            &\leq \frac{\alpha}{\beta}\E{T_\srpt}
                + K(\alpha, \beta) \gp*{\frac{1}{\rho} \ln\frac{1}{1 - \rho} - 1} \E{S},
        \end{align*}
        where
        \begin{equation*}
            K(\alpha, \beta)
            = \gp*{\frac{3}{2}\alpha \1(\beta < 1) + 1}
                \min\curlgp*{1, \max\curlgp[\Big]{1 - \frac{1}{\alpha}, \frac{1}{\beta} - 1}}
            \leq 2.5 \frac{\alpha}{\beta}.
        \end{equation*}
    \item
        The approximation ratio of SRPT-B is at most~$3.5 \alpha/\beta$.
    \item
        As $\alpha$ and~$\beta$ converge to~$1$,
        the approximation ratio of SRPT-B converges to~$1$
        uniformly in the arrival rate
        and the joint distribution of true and estimated sizes.
    \end{enumerate}
\end{restatable*:theorem}

The expression we give for $K(\alpha, \beta)$ in \cref{thm:srpt-b_result}(a) is a compromise between simplicity and tightness.

Intuitively, the key properties of SRPT-B that allow us to prove \cref{thm:srpt-b_result} are:
\begin{itemize}
    \item The rank function is equal to $z-a$ for $a \le z$.
    This ensures that if $z = s$, our policy assigns jobs the same rank as SRPT.
    \item No job receives more than $2r$ service with rank $\le r$,
    for any rank $r < z$.
    Jobs with poor estimates consequently don't cause long delays to low-rank jobs.
    The SRPT-E policy does not have the same guarantee,
    resulting in the poor response times in \cref{thm:srpt-e_result}.
    \item A job's rank never grows larger than $z$, its initial rank.
    Without this cap on the rank the approximation ratio could be arbitrarily poor when $\beta < 1/2$,
    as we explain in \cref{rmk:bounce_limit, sec:SRPT-B-without-cap}.
\end{itemize}

Finally, we consider PSJF-E, which we bound relative to PSJF.
We prove PSJF-E is consistent relative to PSJF,
achieving a mean response time ratio of~$\alpha/\beta$.
While this is a weaker consistency result than that of SRPT-B,
PSJF often performs within a few percent of SRPT in practice,
making the result nearly as strong.
In the worst case, PSJF's mean response time
is within a factor of~$1.5$ of SRPT's \cite{wierman_nearly_2005},
so PSJF-E is within a factor of $1.5 \alpha/\beta$ of optimal.\footnote{%
    We are not aware of any tight examples in which $\E{T_\psjf} = 1.5 \E{T_\srpt}$.
    The largest mean response time ratio between PSJF and SRPT
    that we are aware of is approximately~$1.125$,
    which occurs when $S$ has a Pareto distribution with shape parameter~$1.5$
    and $\rho$ is nearly~$1$.
    A tighter bound on the approximation ratio of PSJF
    would strengthen our bound on the approximation ratio on PSJF-E.}
This is a stronger graceful-degradation bound than we obtained for SRPT-B.

\begin{restatable*:theorem}[Performance of PSJF-E]
    \label{thm:psjf-e_result}
    Consider the M/G/1 with $(\beta, \alpha)$-bounded size estimates.
    \begin{enumerate}[(a)]
    \item
        The mean response time of PSJF-E is bounded above by
        \begin{equation*}
            \E{T_\psjfe} \leq \frac{\alpha}{\beta} \E{T_\psjf}.
        \end{equation*}
    \item
        The approximation ratio of PSJF-E is at most $1.5 \alpha/\beta$.
    \end{enumerate}
\end{restatable*:theorem}

One can view our PSJF-E result
as bounding what \citet{mitzenmacher_scheduling_2020}
dubs the ``price of misprediction'' of PSJF-E.
An algorithm's price of misprediction is its performance ratio
relative not to the optimal algorithm,
in this case SRPT,
but relative to a version of the algorithm that has perfect information,
in this case PSJF.

\subsection{Discussion of Our Results}

In the wider context of online algorithms with predictions,
the three goals discussed in this paper can be stated more generally:

\begin{description}
    \item[Consistency:] In the limit as the prediction quality becomes perfect,
    the performance should approach that of the optimal algorithm with perfect information.  For instance, one might try to achieve $A$-consistency \cite{DBLP:journals/jacm/LykourisV21},
    which requires that the competitive ratio is bounded by $A$ as the error in the predictions goes to 0.
    \item[Robustness:] In the limit as the prediction quality becomes arbitrarily poor, the performance should be comparable to that of the optimal algorithm with perfect information. For instance, one might try to achieve $B$-robustness \cite{DBLP:journals/jacm/LykourisV21}, which requires that the competitive ratio is bounded by $B$ under arbitrarily poor predictions.
    \item[Graceful Degradation:]
    As the prediction quality worsens from perfect to worthless,
    the performance should degrade smoothly and slowly.
    For instance, one might try to achieve $C$-graceful degradation,
    which requires that the competitive ratio is bounded by $C$
    times some measure of the estimate quality.
\end{description}

Looked at in this framework,
we prove that SRPT-B is $1$-consistent and has $3.5$-graceful degradation,
where $\alpha/\beta$ is our measure of estimate quality for $(\beta,\alpha)$-bounded size estimates.
We also prove that PSJF-E is $1.5$ consistent and has $1.5$-graceful degradation,
where the factor of $1.5$ comes from the maximum gap between PSJF and SRPT.

While we do not have robustness results for these algorithms,
that is because in the context of scheduling in the M/G/1,
the robustness property is provably unachievable.
In particular, no online policy without prediction information can achieve
a constant approximation ratio against SRPT,
as discussed in \cref{sec:gittins_bound}.

We feel our emphasis on graceful degradation is a key contribution of our work that may apply to many other algorithms-with-predictions problems.
While consistency and robustness are well-known goals in the literature,
the graceful degradation goal has received less focus.
However, we argue that graceful degradation is \emph{extremely important}.
Real applications often have high-quality but imperfect predictions,
which is the regime where performance is bounded by a graceful degradation result.
The extreme cases of perfect or worthless estimates may come up less in practice.

To adapt the notion of robustness to the setting of M/G/1 scheduling,
one possible method would be to compare an algorithm against the optimal blind policy, which knows the job size distribution, but not the job sizes.
This policy is known; it is called the Gittins index policy \cite{gittins_bandit_1979}.
We discuss this policy in \cref{sec:gittins_bound}.
This method of comparing against the optimal blind policy might be applicable to other algorithms-with-predictions problems.

\section{Proof Overview}
\label{sec:proof_overview}

We now explain the main ideas we use to prove our main results.
We focus on our analysis of SRPT-B (the most complex result),
but we briefly comment on how the same ideas apply to analyzing
SRPT-E (see \cref{rmk:srpt-e_idea}) and PSJF-E (see \cref{rmk:psjf-e_idea}).

Our overall approach to comparing SRPT-B to SRPT
is to compare each to a third policy, namely SRPT-SE.
This approach proved useful because SRPT-SE's rank function
has similarities with both SRPT's and SRPT-B's.
\begin{itemize}
    \item Under both SRPT-SE and SRPT,
    a job's rank at every age is
    within a constant factor of its remaining size.
    \item Under both SRPT-SE and SRPT-B,
    a job's initial rank is its estimated size,
    and a job's rank never exceeds its initial rank.
\end{itemize}

In the remainder of this section,
we explain how the above properties help us compare SRPT-SE
to each of SRPT and SRPT-B.
Interestingly, the two comparisons make use of two very different methods
of analyzing mean response time.

\subsection{Comparing SRPT-SE to SRPT}
\label{sec:proof_overview:work_integral}

SRPT minimizes mean response by prioritizing jobs by remaining size
\cite{schrage_proof_1968}.
SRPT-SE almost prioritizes jobs by remaining size,
but it can make an error whenever two jobs' remaining sizes
are within a constant factor of each other.
The specific factor is $\alpha/\beta$:
if job~$1$ has remaining size~$r_1$
and job~$2$ has greater remaining size $r_2 > r_1$,
then SRPT will always serve job~$1$,
but SRPT-SE might serve job~$2$ if $\beta r_2 \leq \alpha r_1$.

Intuitively,
one might hope that because SRPT-SE only makes constant-factor errors
when prioritizing jobs,
its mean response time should suffer by only a constant factor.
We show that this indeed is the case,
and that the constant factor is the same.
Specifically, \cref{thm:srpt-se_result} states
\begin{equation}
    \label{eq:srpt-se_result_informal}
    \E{T_\srptse} \leq \frac{\alpha}{\beta} \E{T_\srpt}.
\end{equation}

In order to show~\cref{eq:srpt-se_result_informal},
we use a very recently developed formula for mean response time
\cite{scully_gittins_2020}.
At a high level, the formula expresses a policy's mean response time
in terms of an integral of various types of work.
We first describe the formula
(see \cref{sec:proof_overview:work_integral:formula})
and then describe how we apply it to proving~\cref{eq:srpt-se_result_informal}
(see \cref{sec:proof_overview:work_integral:application})

\subsubsection{Mean Response Time as a Work Integral}
\label{sec:proof_overview:work_integral:formula}

Define the \emph{$(\remsize \leq r)$-work} of a system to be
the total remaining size of the jobs in the system
that have remaining size $r$ or less,
as illustrated in \cref{fig:remsize_leq_r-work}.
\citet[Theorem~6.3]{scully_gittins_2020} show that
we can write the mean response time of any policy~$\generic$ in terms of
the mean amount of $(\remsize \leq r)$-work in the system:
\begin{equation}
    \label{eq:work_integral_informal}
    \E{T_\generic} =
    \frac{1}{\lambda} \int_0^\infty \frac{\E{\text{$(\remsize \leq r)$-work under~$\generic$}}}{r^2} \d{r}.
\end{equation}
See \cref{def:phi-work, def:shorthand} for
a formal definition of $(\remsize \leq r)$-work
and \cref{thm:response_time_work_integral} for
a formal statement of~\cref{eq:work_integral_informal}.

\begin{figure}
    \nosubfigure
    \begin{tikzpicture}[figure, x=16.5, y=16.5]
    \draw (2, 0) circle (2);
    \draw[thick, ->] (4, 0) -- ++(0.8, 0);
    \draw (-8.8, -2) -- ++(8.8, 0) -- ++(0, 4) -- ++(-8.8, 0);
    \draw (0, -2) -- ++(0, 4);
    \draw (-2, -2) -- ++(0, 4);
    \draw (-4, -2) -- ++(0, 4);
    \draw (-6, -2) -- ++(0, 4);
    \draw (-8, -2) -- ++(0, 4);
    \draw[thick, ->] (-10, 0) 
      -- ++(0.8, 0);
    \node at (-10, 2) {};
    \newcommand{\job}[3]{%
      \fill[cyan!51!white] (#1 - 0.4, -1.5) rectangle ++(0.8, #2/10 - #3/10);
      \fill[orange!28!white] (#1 - 0.4, -1.5 + #2/10 - #3/10) rectangle ++(0.8, #3/10);
      \ifstrequal{#3}{#2}{}{\draw[cyan!48!black] (#1 - 0.4, -1.5 + #2/10 - #3/10) -- ++(0.8, 0);}
      \draw[orange!48!black] (#1 - 0.4, -1.5) rectangle ++(0.8, #2/10);
    }
    \newcommand{\jobAnnotated}[5]{%
      \job{#2}{#3}{#4}
      \node[below,align=center] at (#2, -1.85)
        {\strut $r_{#1} = #4$};
      \ifstrequal{#5}{>}{
        \draw[ultra thick, red!68!blue!82!white]
          (#2, -3.5) ++(-0.25, -0.25) -- ++(0.5, 0.5) ++(0, -0.5) -- ++(-0.5, 0.5);
      }{
        \draw[ultra thick, green!48!yellow!65!black]
          (#2, -3.5) ++(-0.35, -0.05) -- ++(0.2, -0.2) -- ++(0.5, 0.5);
      }
    }
    \jobAnnotated{1}{2}{20}{5}{\leq} 
    \jobAnnotated{2}{-1}{17}{7}{\leq} 
    \jobAnnotated{3}{-3}{11}{9}{\leq} 
    \jobAnnotated{4}{-5}{27}{14}{>} 
    \jobAnnotated{5}{-7}{24}{24}{>} 
    \draw[decorate, decoration={brace, raise=0.25em}]
      (3 - 1/16, -4) -- (-4 + 1/16, -4)
      node[midway, below=0.4em, align=center]
      {$\text{$(\remsize \leq 10)$-work} = r_1 + r_2 + r_3 = 21$};
    \begin{scope}[shift={(7.6, 0)}]
      \fill[black!14, rounded corners=0.5em] (-2.2, -3.2) rectangle (5.6, 2) node {};
      \job{0}{30}{18}
      \draw[decorate, decoration={brace, raise=0.25em}]
        (-0.4, -1.5 + 1/16) -- ++(0, 3 - 1/8)
        node[midway, left=0.4em] {\strut size};
      \draw[decorate, decoration={brace, mirror, raise=0.25em}]
        (0.4, -1.5 + 1/16) -- ++(0, 1.2 - 1/8)
        node[midway, right=0.4em] {\strut age};
      \draw[decorate, decoration={brace, mirror, raise=0.25em}]
        (0.4, -1.5 + 1.2 + 1/16) -- ++(0, 3 - 1.2 - 1/8)
        node[midway, right=0.4em] {\strut remaining size~$r_j$};
      \node[below] at (1.85, -1.7) {\strut \textbf{\textsf{Legend}} (jobs indexed by~$j$)};
    \end{scope}
\end{tikzpicture}
    \caption{Example of $(\remsize \leq r)$-Work}
    \label{fig:remsize_leq_r-work}
\end{figure}

\subsubsection{Comparing SRPT-SE's and SRPT's Work Integrals}
\label{sec:proof_overview:work_integral:application}

With \cref{eq:work_integral_informal} in hand,
to compare the mean response times of SRPT-SE and SRPT,
it suffices to compare their amounts of $(\remsize \leq r)$-work.
In the proof of \cref{thm:srpt-se_result}, we show
\begin{equation}
    \label{eq:srpt-se_key_step}
    \E{\text{$(\remsize \leq r)$-work under SRPT-SE}}
    \leq \E[\bigg]{\text{$\gp[\bigg]{\remsize \leq \frac{\alpha}{\beta} r}$-work under SRPT}}.
\end{equation}
Combining \cref{eq:work_integral_informal, eq:srpt-se_key_step}
and using a change of variables
implies~\cref{eq:srpt-se_result_informal}.

The intuition behind \cref{eq:srpt-se_key_step} is as follows.
Because SRPT always serves the job of least remaining size,
it satisfies the following guarantee:
\begin{quote}
    Whenever the system has nonzero $(\remsize \leq r)$-work,
    SRPT serves a job of remaining size $r$ or less,
    thus decreasing the amount of $(\remsize \leq r)$-work.
\end{quote}
This guarantee implies that SRPT minimizes mean $(\remsize \leq r)$-work
among all scheduling policies.
In contrast, SRPT-SE satisfies a weaker guarantee:
\begin{quote}
    Whenever the system has nonzero $(\remsize \leq r)$-work,
    SRPT-SE serves a job of remaining size $\alpha/\beta \cdot r$ or less,
    thus decreasing the amount of $(\remsize \leq \alpha/\beta \cdot r)$-work.
\end{quote}
Roughly speaking, this means that
whenever the system's $(\remsize \leq r)$-work is nonzero,
SRPT-SE reduces $(\remsize \leq \alpha/\beta \cdot r)$-work
just as efficiently as SRPT does,
suggesting a relationship like \cref{eq:srpt-se_key_step} might hold.

The main technical challenge in proving \cref{eq:srpt-se_key_step}
is formalizing the above intuition.
The key ingredient turns out to be introducing
a new variant of $(\remsize \leq r)$-work.
The new variant, called \emph{$(\remsizee \leq r)$-work}
(see \cref{def:phi-work, def:shorthand}),
uses scaled estimated remaining size instead of true remaining size.
This new variant is important because
SRPT-SE always serves the job of least scaled estimated remaining size,
so it satisfies the following guarantee:
\begin{quote}
    Whenever the system has nonzero $(\remsizee \leq r)$-work,
    SRPT-SE serves a job of scaled estimated remaining size $r$ or less,
    thus decreasing the amount of $(\remsizee \leq r)$-work.
\end{quote}
This guarantee implies that,
analogously to SRPT minimizing mean $(\remsize \leq r)$-work,
SRPT-SE minimizes mean $(\remsizee \leq r)$-work
(see \cref{thm:soap_match_minimizes_work}).

\begin{remark}
    \label{rmk:psjf-e_idea}
    The proof of \cref{thm:psjf-e_result},
    which compares PSJF-E to PSJF,
    follows a similar strategy to the comparison of SRPT-SE to SRPT
    outlined above.
    However, the details are significantly more complicated.
    The main obstacle is that
    while \cref{eq:work_integral_informal} uses remaining size,
    PSJF-E and PSJF prioritize jobs by \emph{original}
    estimated and true size, respectively.
    We overcome this obstacle by introducing
    more variants of $(\remsize \leq r)$-work.
\end{remark}

\subsection{Comparing SRPT-B to SRPT-SE}
\label{sec:proof_overview:soap}

Our approach to comparing SRPT-B to SRPT-SE
looks very different from our approach to comparing SRPT-SE to SRPT.
In particular, we use a different method
of characterizing each policy's mean response time.
The method, often called the ``tagged-job method'',
has been used since the early days of M/G/1 scheduling theory
to analyze a variety of policies,
including SRPT \cite{schrage_queue_1966}.
Recently, \citet{scully_soap_2018} generalized the tagged-job method
to \emph{all} policies in which a job's rank varies as a function of its age,
including all of the policies we study (see \cref{def:policy}).
Below, we outline how the tagged-job method of \citet{scully_soap_2018}
applies to SRPT-B and SRPT-SE.

At a high level, the tagged-job method works
by following a single ``tagged'' job on its journey through the system.
The tagged job's response time is a random variable
with several sources of randomness:
\begin{itemize}
    \item the random true size~$S$ and estimated size~$Z$ of the tagged job,
    \item the random state of the system at the moment the tagged job arrives,
    and
    \item the random arrivals that occur after the tagged job.
\end{itemize}
Using the fact that arrival times are Poisson \cite{wolff_poisson_1982},
one can show that the expected response time of the tagged job,
where the expectation is taken over all of the above sources of randomness,
is indeed the system's mean response time
\cite{harchol-balter_performance_2013}.

To compute the tagged job's expected response time,
we first condition on its true and estimated sizes.
Specifically, let the random variable $T_\generic(s, z)$ denote
the response time of the tagged job under policy~$\generic$
given that it has true size $S = s$ and estimated size $Z = z$.
We will find $\E{T_\generic(s, z)}$,
from which mean response time follows by integrating over $s$ and~$z$:
\begin{equation}
    \label{eq:tagged_to_mean}
    \E{T_\generic}
    = \int_0^\infty \nestint_0^\infty \E{T_\generic(s, z)} f_{S, Z}(s, z) \d{s} \d{z}.
\end{equation}

To analyze the tagged job's response time~$T_\generic(s, z)$,
we split it into two parts:
\begin{description}
\item[Waiting time:]
    the amount of time between the tagged job's arrival and
    the moment the tagged job first receives service,
    denoted~$\Twait_\generic(s, z)$.
\item[Residence time:]
    the amount of time between the tagged job first receives service and
    the tagged job's completion,
    denoted~$\Tres_\generic(s, z)$.
\end{description}
We illustrate waiting time and residence time in \cref{fig:waiting_residence}.
We define mean waiting and residence times $\E{\Twait_\generic}$ and $\E{\Tres_\generic}$
analogously to~\cref{eq:tagged_to_mean}.

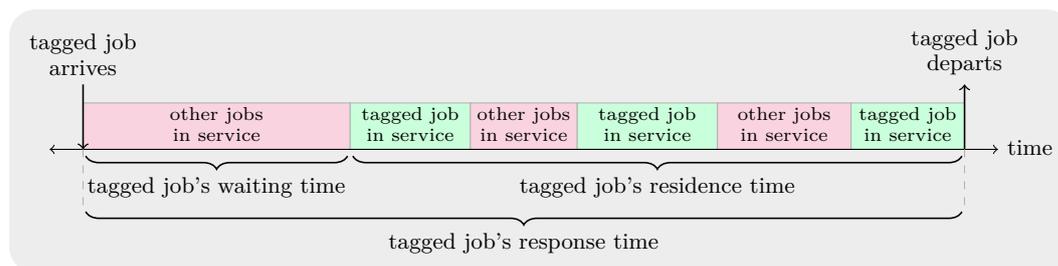
\begin{figure}
    \nosubfigure
    \renewcommand{\xscale}{10}
\begin{tikzpicture}[figure]
    \newcommand{\heightService}{2.5}
    \newcommand{\serveTagged}[2]{
      \draw[guide, solid, fill=green!70!cyan!21]
        (#1, 0) rectangle (#2, \heightService)
        node[black, midway, align=center, font=\scriptsize]
        {\vphantom{by}tagged job\\ in service\vphantom{by}};
    }
    \newcommand{\serveOther}[2]{
      \draw[guide, solid, fill=magenta!90!red!21]
        (#1, 0) rectangle (#2, \heightService)
        node[black, midway, align=center, font=\scriptsize]
        {\vphantom{by}other jobs\\ in service\vphantom{by}};
    }
    
    \newcommand{\xstart}{3}
    \newcommand{\xextra}{1.25}
    
    \serveOther{\xstart}{13}
    \serveTagged{13}{17.5}
    \serveOther{17.5}{21.5}
    \serveTagged{21.5}{26.75}
    \serveOther{26.75}{31.75}
    \serveTagged{31.75}{36}
    
    \draw[axis, <->] (\xstart - \xextra, 0) node {} -- (36 + \xextra, 0) node[right] {\vphantom{by}time};
    \draw[<-] (\xstart, 0) -- ++(0, 3.5) node[black, above, align=center] {\vphantom{by}tagged job\\arrives};
    \draw[->] (36, 0) -- ++(0, 3.5) node[black, above, align=center] {\vphantom{by}tagged job\\departs};
    
    \draw[guide] (\xstart, -3) -- ++(0, 3) (36, -3) -- ++(0, 3);
    \draw[decorate, decoration={brace, raise=0.25em, amplitude=0.5em, mirror}]
      (\xstart + 0.1, -3) -- (36 - 0.1, -3)
      node[midway, below=0.7em]
      {\vphantom{by}tagged job's response time};
    \draw[decorate, decoration={brace, raise=0.25em, amplitude=0.5em, mirror}]
      (\xstart + 0.1, 0) -- (13 - 0.1, 0)
      node[midway, below=0.7em]
      {\vphantom{by}tagged job's waiting time};
    \draw[decorate, decoration={brace, raise=0.25em, amplitude=0.5em, mirror}]
      (13 + 0.1, 0) -- (36 - 0.1, 0)
      node[midway, below=0.7em]
      {\vphantom{by}tagged job's residence time};
\end{tikzpicture}
    \caption{Response Time${} = {}$Waiting Time${} + {}$Residence Time}
    \label{fig:waiting_residence}
\end{figure}

To compare SRPT-B to SRPT-SE,
we separately compare the policies' waiting times
(see \cref{sec:proof_overview:soap:waiting, thm:srpt-b_waiting})
and residence times
(see \cref{sec:proof_overview:soap:residence, thm:srpt-b_residence}).
At a high level,
because SRPT-B and SRPT-SE have similar enough rank functions,
we are able to show that SRPT-B's waiting and residence times
are not too much larger than SRPT-SE's.

\subsubsection{Comparing Waiting Times}
\label{sec:proof_overview:soap:waiting}

Consider a tagged job of estimated size~$z$.
Under both SRPT-B and SRPT-SE, the tagged job's initial rank is~$z$.
The tagged job's waiting time therefore lasts until
any other jobs that remain in the system have rank greater than~$z$,
at which point the tagged job,
having better rank than all other jobs in the system,
is served for the first time.
Therefore, the tagged job's waiting time depends on
how long each \emph{other} job spends with rank $z$ or less.
In particular, the tagged job's waiting time does not depend on its own size,
so we denote its waiting time by simply $\Twait_\generic(z)$.
Specifically, letting
\begin{equation*}
    u_\generic(z) = \E[\bigg]{\gp[\bigg]{\begin{tabular}{@{}l@{}}\text{amount of service time during which} \\ \text{a job has rank $z$ or less under policy~$\generic$}\end{tabular}}^2},
\end{equation*}
it turns out that comparing $\E{\Twait_\srptb(z)}$ to $\E{\Twait_\srptse(z)}$
boils down to comparing $u_\srptb(z)$ to $u_\srptse(z)$
(see \cref{thm:soap}(a)).

One can use simple geometry to show that
under SRPT-B, the amount of service time a job spends with rank $z$ or less
is at most than twice the amount it would be under SRPT-SE,
implying $u_\srptb(z) \leq 4 u_\srptse(z)$.
This is strong enough to show graceful degradation of SRPT-B,
but it does not imply consistency.
But the rank functions of SRPT-B and SRPT-SE do become closer and closer together
as $\alpha$ and $\beta$ approach~1,
so one would expect consistency of SRPT-B to hold, too.

The main technical challenge in comparing waiting times
is obtaining a bound tight enough to show consistency.
In particular,
it does not suffice to simply bound $u_\srptb(z) - u_\srptse(z)$
with a quantity that vanishes as $\alpha$ and $\beta$ approach~1.
While this is a necessary first step,
it shows only that as $\alpha$ and $\beta$ approach~1,
the difference $\E{\Twait_\srptb(z)} - \E{\Twait_\srptse(z)}$
vanishes for all~$z$.
We seek bounds on mean response time,
so we need to integrate over~$z$ to show that
$\E{\Twait_\srptb} - \E{\Twait_\srptse}$ also vanishes.
This second step is purely computational but requires some care:
there are several choices one must make when bounding the integral,
and many choices lead to either intractable expressions
or bounds that are too weak to show consistency.

\begin{remark}
    \label{rmk:srpt-e_idea}
    The reason for SRPT-E's poor performance is that under SRPT-E,
    a job can spend up to
    a $1 - \beta$ fraction of its service time below rank~$0$.
    This means one can have $u_\srpte(z) \geq (1 - \beta)\E{S^2}$,
    from which \cref{thm:srpt-e_result} easily follows.
    SRPT-B avoids this problem thanks to the bounce in its rank function.
\end{remark}

\subsubsection{Comparing Residence Times}
\label{sec:proof_overview:soap:residence}

Consider a tagged job of true size~$s$ and estimated size~$z$.
When the tagged job starts its residence time,
its rank~$z$ is less than the rank of every other job in the system.
Moreover, under both SRPT-B and SRPT-SE,
a job's rank never exceeds this initial rank of~$z$
(see \cref{fig:rank:srpt-b, fig:rank:srpt-se}).
Therefore, the only reason that the tagged job might be preempted
is if new jobs arrive.

Suppose a new job of estimated size~$z'$ arrives
while the tagged job has age~$a$.
What determines whether the new arrival delays the tagged job?
\begin{itemize}
    \item If the new job's initial rank~$z'$
    is less than the tagged job's rank at age~$a$,
    then the new job has priority over the tagged job.
    \item If $z'$ is at least the tagged job's rank at age~$a$,
    then the tagged job has priority over the new job, initially. But if later the tagged job will have rank greater than~$z'$
    at some \emph{future} age $a' > a$,
    then when the tagged job reaches age~$a'$,
    the new job will have priority over the tagged job.
\end{itemize}
The conclusion of this discussion is what \citet{scully_soap_2018}
call the ``Pessimism Principle'',
the upshot of which is the following:
\begin{quote}
    When determining whether a new arrival will delay the tagged job,
    what matters is not the tagged job's current rank
    but rather its \emph{worst future rank}.
\end{quote}
We illustrate the difference between a job's rank and worst future rank
under SRPT-B in \cref{fig:worst_future_rank}.
Note that a job's worst future rank depends
not just on its estimated size and age
but also on its true size.

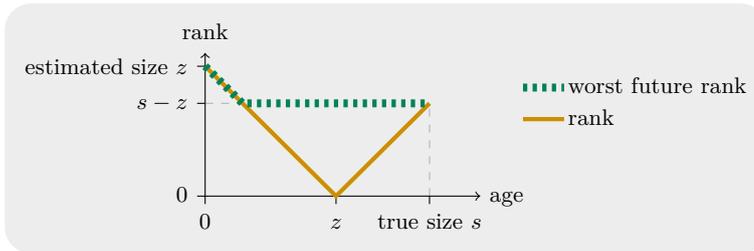
\begin{figure}
    \nosubfigure
    \begin{tikzpicture}[figure]
    \newcommand{\z}{7}
    \newcommand{\s}{12}
    \newcommand{\diffexpr}{(\s-\z)*\s/\z}
    \newcommand{\diff}{{\diffexpr}}
    
    \yguide[estimated size $z$]{0}{\z}
    \xguide[$z$]{\z}{0}
    \xguide[true size $s$]{\s}{\s-\z}
    \yguide[$s - z$]{{\z-(\s-\z)}}{\s-\z}
    
    \axes{14}{7}{$0$}{age}{$0$}{rank}
    \draw[srpt-b] (0, \z) -- ++(\z, -\z) -- ++({\s-\z}, {\s-\z});
    \draw[psjf, dotted, line width=3pt] (0, \z) -- ++({\s-2*(\s-\z)}, {-(\s-2*(\s-\z))}) -- ++({2*(\s-\z)}, 0);
    
    \begin{scope}[shift={(17, 5.85)}]
        \newcommand{\spacing}{1.7}
        \draw[psjf, dotted, line width=3pt] (0, 0) -- ++(2.2, 0);
        \node[right] at (1.9, 0) {\vphantom{by}worst future rank};
        \draw[srpt-b] (0, -\spacing) -- ++(2.2, 0);
        \node[right] at (1.9, -\spacing) {\vphantom{by}rank};
    \end{scope}
\end{tikzpicture}
    \caption{Example of Worst Future Rank under SRPT-B}
    \label{fig:worst_future_rank}
\end{figure}

The Pessimism Principle means that
bounding the tagged job's residence time of SRPT-B
boils down to bounding the tagged job's worst future rank under SRPT-B.
One simple bound on the tagged job's worst future rank is
its initial rank~$z$,
because under SRPT-B, a job's rank never exceeds its initial rank
(see \cref{fig:rank:srpt-b}).
As it happens, under PSJF-E, the tagged job's worst future rank
would always be~$z$.
This means that SRPT-B's mean residence time is at most that of PSJF-E,
which turns out to be simple to bound
(see \cref{thm:psjf-e_residence, thm:srpt_log}).
This is strong enough to show graceful degradation of SRPT-B,
but it does not imply consistency.

The main technical challenge in comparing residence times
is obtaining a bound tight enough to show consistency.
To show consistency of SRPT-B,
we would like to bound SRPT-B's residence time
in terms of that of SRPT-SE, not PSJF-E.
However, the tagged job's worst future rank at a given age
can be greater under SRPT-B than under SRPT-SE.
Our solution is, roughly speaking,
to bound the worst future rank under SRPT-B
to a ``shifted'' version of worst future rank under SRPT-SE
(see \cref{fig:residence}).
The result is a bound of the form
\begin{equation*}
    \E{\Tres_\srptb(s, z)}
    \leq \E{\Tres_\srptse(s, z)} + c \E{\Tres_\psjfe(s, z)},
\end{equation*}
where $c$ approaches~0 as $\alpha$ and $\beta$ approach~1.
This immediately implies an analogous bound on mean residence times.

\begin{remark}
    \label{rmk:bounce_limit}
    The Pessimism Principle,
    namely the fact that a job's residence time is governed by
    its worst future rank instead of its current rank,
    is the reason we cap the bounce in SRPT-B's rank function
    (see \cref{fig:rank:srpt-b})
    to no more than the job's initial rank.
    \Cref{sec:SRPT-B-without-cap} explains in more detail
    how performance degrades without this rank cap.
\end{remark}

\section{Background on M/G/1 Scheduling Theory}
\label{sec:background}

In this section, we review definitions and results from M/G/1 scheduling theory
that we use in our proofs.
Specifically, we review two recently developed methods for
computing a policy's mean response time.
\begin{itemize}
\item
    \Cref{sec:background:work_integral} reviews the ``work integral'' method,
    which we use to compare SRPT-SE to SRPT.
\item
    \Cref{sec:background:tagged_job} reviews the ``tagged job'' method,
    which we use to compare SRPT-B to SRPT-SE.
\end{itemize}
Having given in \cref{sec:proof_overview} an intuitive overview of each method,
the main purpose of this section is to present them more formally.

\subsection{Mean Response Time via the Work Integral Method}
\label{sec:background:work_integral}

We saw in \cref{sec:proof_overview:work_integral} that
one can compute a policy's mean response time
by looking at the amount of different types of work in the system.
Specifically, the key definition is $(\remsize \leq r)$-work,
the amount of work contributed by jobs
which have remaining size $r$ or less.

Below, we give a formal definition of a general kind of work,
which includes $(\remsize \leq r)$-work as a special case.
Recall from \cref{def:policy} that a job's state is
a tuple $x = (s, z, a)$
consisting of its true size~$s$, estimated size~$z$, and age~$a$.

\begin{definition}
    \label{def:phi-work}
    Let $\phi : \R_+^3 \to \{\mathsf{false}, \mathsf{true}\}$
    be a predicate on job states.
    \begin{enumerate}[(a)]
    \item
        The \emph{$\phi$-work of a job in state~$x$}, denoted $w(x, \phi)$,
        is the amount of service a job in state~$x$ requires to either
        complete or reach a state that does not satisfy~$\phi$.
        That is, a job's $\phi$-work, roughly speaking,
        is its remaining processing time while satisfying~$\phi$.
        Formally,
        \begin{equation*}
            w((s, z, a), \phi) = \sup\curlgp{w \in [0, s - a) \given \phi(s, z, a + w)}.
        \end{equation*}
    \item
        The \emph{(system) $\phi$-work} is the total
        $\phi$-work of all jobs in the system.
        We denote by $W_\generic(\phi)$ the steady-state distribution of
        the system $\phi$-work under policy~$\generic$.
    \end{enumerate}
\end{definition}

When discussing $\phi$-work,
it is helpful to have a shorthand notation for describing predicates.
The following definition describes such a shorthand.

\begin{definition}
    \label{def:shorthand}
    \leavevmode
    \begin{enumerate}[(a)]
    \item
        Let $\func : \R_+^3 \to \R$ be a real function on job states
        and $r \in \R$ be a constant.
        The predicate $(\func \leq r)$ is true for those states~$x$
        such that $\func(x) \leq r$.
        We define other inequality predicates similarly,
        e.g. $\func_1 \leq r < \func_2\esub$,
        and we omit the parentheses when they would be redundant,
        e.g. $W_\generic(\func \leq r)$.
        To disambiguate between functions and constants,
        we use $\mathsf{sans\mathhyphen serif}$ font for functions.
    \item
        We frequently use the following functions on job states
        in the above shorthand:
        {\setlength{\jot}{0em}
        \begin{alignat*}{3}
            & \text{\emph{(true) size}} \qquad &
            & \size(s, z, a) & &= s, \\
            & \text{\emph{(true) remaining size}} \qquad &
            & \remsize(s, z, a) & &= s - a, \\
            & \text{\emph{estimated size}} \qquad &
            & \sizee(s, z, a) & &= z, \\
            & \text{\emph{scaled estimated remaining size}} \qquad &
            & \remsizee(s, z, a) & &= z/s \cdot (s - a).
        \end{alignat*}}
    \end{enumerate}
\end{definition}

All of the functions in \cref{def:shorthand}(b)
happen to also be rank functions of one of the policies we study
(see \cref{def:policy}).
For example, $\remsize = \rank{\srpt}$.
We introduce the names in \cref{def:shorthand}(b)
to emphasize that, for instance,
we can consider $(\remsize \leq r)$-work under policies other than SRPT.

This last example is especially important,
because a policy's mean response time is connected to
its steady-state $(\remsize \leq r)$-work.

\begin{proposition}[\textnormal{%
    special case of \cite[Theorem~6.3]{scully_gittins_2020}}]
    \label{thm:response_time_work_integral}
    In the M/G/1, the mean response time of any policy~$\generic$ is
    \begin{equation*}
        \E{T_\generic}
        = \frac{1}{\lambda} \int_0^\infty \frac{\E{W_\generic(\remsize \leq r)}}{r^2} \d{r}.
    \end{equation*}
\end{proposition}

The above result is a recently derived, powerful identity
for the mean response time.
We use it in our analyses of SRPT-SE (see \cref{sec:srpt-se})
and PSJF-E (see \cref{sec:psjf-e}).

\subsection{Mean Response Time via the Tagged Job Method}
\label{sec:background:tagged_job}

We describe the main ideas behind the tagged job method
in \cref{sec:proof_overview:soap}.
As a reminder, the approach is to focus on a single ``tagged'' job;
split its response time into two parts,
\emph{waiting time} and \emph{residence time} (see \cref{fig:waiting_residence});
and analyze each part separately.
The purpose of this section is to define the concepts and notation
that we need in order to write down formulas
for the tagged job's expected waiting and residence times.

Consider a tagged job of of true size~$s$ and estimated size~$z$
arriving to a steady-state system under some scheduling policy~$\generic$.
An important quantity when computing the tagged job's waiting and residence times
is the rate at which jobs with rank less than the tagged job
arrive to the system.
We can interpret the system load $\rho = \lambda \E{S}$
as the overall rate at which work arrives.
Analogously, define\footnote{%
    Recall that a random job's true size~$S$ and estimated size~$Z$
    are \emph{not} independent,
    which is important in the definition of $\rho_Z(z)$.}
\begin{align}
    \label{eq:rho}
    \rho_S(s) &= \lambda \E{S \1(S \leq s)}, &
    \rho_Z(z) &= \lambda \E{S \1(Z \leq z)}.
\end{align}
These are the average rates at which work arrives
when one only counts work from jobs
whose true size or estimated size, respectively,
is at most some threshold.
Specifically, $\rho_S(s)$ is important for analyzing
policies that use a job's true size (SRPT and PSJF),
while $\rho_Z(z)$ is important for analyzing
policies that use a job's estimated size (SRPT-E, PSJF-E, SRPT-B, and SRPT-SE).

Having defined \cref{eq:rho}, there are two more quantities we need to define
before stating formulas for the tagged job's waiting and residence times.
We give formal and informal expressions for each.
\begin{itemize}
\item
    In the waiting time formula, we use the quantity\footnote{%
        In the formal expression, $\vert{\cdot}$ denotes interval length.
        The definition we give is simplified by the fact that
        for the scheduling policies we consider,
        a job's rank is below a threshold~$r$
        for at most one contiguous interval of ages.
        A more complicated definition is needed for policies
        with general rank functions \cite{scully_soap_2018}.}
    \begin{align*}
        u_\generic(r)
        &= \E[\big]{\vert[\big]{\curlgp{a \in [0, S) \given \rank{\generic}(S, Z, a) \leq r}}^2} \\
        &= \E[\bigg]{\gp[\bigg]{\begin{tabular}{@{}l@{}}\text{amount of service time during which} \\ \text{a job has rank $r$ or less under policy~$\generic$}\end{tabular}}^2}.
    \end{align*}
\item
    In the residence time formula, we use the \emph{worst future rank} of a job,
    which we define as
    \begin{align*}
        \worst{\generic}(s, z, a)
        &= \sup_{b \in [a, s)} \rank{\generic}(s, z, b) \\
        &= \gp[\bigg]{\begin{tabular}{@{}l@{}}\text{maximum rank a job currently in state $(s, z, a)$ has} \\ \text{under policy~$\generic$ between now and its completion}\end{tabular}}.
    \end{align*}
\end{itemize}

We are now ready to state the waiting and residence time formulas
for the policies we consider.

\begin{proposition}[\textnormal{%
    special case of \cite[Theorem~5.5]{scully_soap_2018}}]
    \label{thm:soap}
    Consider an M/G/1 under policy $\generic \in \policies$.
    \begin{enumerate}[(a)]
    \item
        The expected waiting time of a (tagged) job of estimated size~$z$ is
        \begin{equation*}
            \E{\Twait_\generic(z)}
            = \frac{\lambda}{2} \frac{u_\generic(z)}{(1 - \rho_Z(z))^2}.
        \end{equation*}
    \item
        The expected residence time of a (tagged) job of true size~$s$
        and estimated size~$z$ is
        \begin{equation*}
            \E{\Tres_\generic(s, z)}
            = \int_0^s \frac{1}{1 - \rho_Z(\worst{\generic}(s, z, a))} \d{a}.
        \end{equation*}
    \end{enumerate}
\end{proposition}

Some intuition for the formulas above is warranted.
We explain one simple case below,
referring the reader to \citet[Section~4]{scully_soap_2018}
for more discussion.

\begin{example}
    \label{rmk:residence_busy_period}
    Consider the residence time of a job of true size~$s$ and estimated size~$z$ under PSJF-E.
    The job's rank is always~$z$ (see \cref{def:policy, fig:rank:psjf-e}),
    which means $\worst{\psjfe}(s, z, a) = z$.
    Applying \cref{thm:soap}(b) yields.
    \begin{equation}
        \label{eq:psjf-e_residence_s_z}
        \E{\Tres_\psjfe(s, z)} = \frac{s}{1 - \rho_Z(z)}.
    \end{equation}
    The intuitive interpretation of \cref{eq:psjf-e_residence_s_z} is as follows.
    During the tagged job's residence time, new jobs may arrive at any time,
    preempting the job in service if they have rank below~$z$.
    On average, the server spends a $\rho_Z(z)$ fraction of its time
    serving these new jobs of rank below~$z$,
    leaving a $1 - \rho_Z(z)$ fraction for serving the tagged job.
    This means that the tagged job's age increases
    at average rate $1/(1 - \rho_Z(z))$,
    so it takes $s/(1 - \rho_Z(z))$ time to go from age~$0$ to age~$s$.
\end{example}

Formulas very similar to those in \cref{thm:soap} hold for SRPT and PSJF.
In fact, such formulas are classic results
\cite{schrage_queue_1966, harchol-balter_performance_2013}.
One may view the SRPT and PSJF formulas
as special cases of the SRPT-SE and PSJF-E formulas
in a system where $S = Z$ for all jobs,
so, for example, we can write them using $\rho_S(s)$ instead of $\rho_Z(z)$.
We omit the exact statements because in our proofs,
we end up analyzing SRPT and PSJF with the work integral method.

\subsection{Useful Lemmas}

The following simple lemmas will be useful in our later analyses.

\begin{lemma}
    \label{thm:rho_derivative}
    \begin{align*}
        \dd{s} \rho_S(s) &= \lambda s f_S(s), &
        \dd{z} \rho_Z(z) &= \lambda \E{S \given Z = z} f_Z(z).
    \end{align*}
\end{lemma}

\begin{proof}
    These follow from \cref{eq:rho}
    and the fact that we can write
    \begin{align*}
        \E{S \1(S \leq s)} &= \int_0^s s' f_S(s') \d{s'} &
        \E{S \1(Z \leq z)} &= \int_0^z \E{S \given Z = z'} f_Z(z') \d{z'}.
        \qedhere
    \end{align*}
\end{proof}

\begin{lemma}
    \label{thm:psjf-e_residence}
    The mean residence time of PSJF-E is
    \begin{equation*}
        \E{\Tres_\psjfe} = \gp*{\frac{1}{\rho} \ln\frac{1}{1 - \rho}}\E{S}.
    \end{equation*}
\end{lemma}

\begin{proof}
    Combining \cref{eq:psjf-e_residence_s_z, thm:rho_derivative} yields
    \begin{align*}
        \E{\Tres_\psjfe}
        &= \int_0^\infty \nestint_0^\infty \E{\Tres_\psjfe(s, z)} f_{S, Z}(s, z) \d{s} \d{z}
        \by{conditioning on $s$ and~$z$} \\
        &= \int_0^\infty \nestint_0^\infty \frac{s}{1 - \rho_Z(z)} f_{S, Z}(s, z) \d{s} \d{z}
        \by{\cref{eq:psjf-e_residence_s_z}} \\
        &= \int_0^\infty \frac{\E{S \given Z = z}}{1 - \rho_Z(z)} f_Z(z) \d{z} \\
        &= \frac{1}{\lambda} \ln\frac{1}{1 - \rho}.
        \by{\cref{thm:rho_derivative}}
    \end{align*}
    The lemma follows from the fact that $\rho = \lambda \E{S}$.
\end{proof}

The main reason that \cref{thm:psjf-e_residence} is useful
is the following result of \citet{wierman_nearly_2005},
which shows that the mean residence time of PSJF-E is
a \emph{lower bound} on the mean response time of SRPT.

\begin{proposition}[\textnormal{\cite[Theorem~5.8]{wierman_nearly_2005}}]
    \label{thm:srpt_log}
    The mean response time of SRPT is bounded below by
    \begin{equation*}
        \E{T_\srpt} \geq \gp*{\frac{1}{\rho} \ln\frac{1}{1 - \rho}} \E{S}.
    \end{equation*}
\end{proposition}

\section{SRPT with Estimates (SRPT-E)}
\label{sec:srpt-e}

Our first result shows that, for $(\beta,\alpha)$-bounded size estimates with $\beta < 1$, the performance of SRPT-E can lead to arbitrarily large approximation ratios.  This formalizes previous empirical results (see e.g. \cite{mitzenmacher_scheduling_2020}), where it was noted that underestimates of large jobs, particularly when job sizes are highly variable, can lead to poor performance for SRPT-E, as a large underestimated job being served can obtain a  negative estimated remaining time and block service for all other jobs, even when the actual remaining time is large.  This result motivates our seeking a variation of SRPT that avoids this problem, namely SRPT-B, and our examination of PSJF-E, which we show in contrast has bounded approximation ratio for $(\beta,\alpha)$-bounded size estimates.

\restate{thm:srpt-e_result}

\begin{proof}
    For a given job distribution $S$, let $Z = \beta S$.

    From \cref{thm:soap}(a),
    we know that the expected waiting time of SRPT-E
    is
    \begin{align*}
        \E{\Twait_\srpte(z)}
            &= \frac{\lambda}{2} \frac{\E[\Big]{\gp[\Big]{{\footnotesize \begin{tabular}{@{}l@{}} \text{amount of service time during which} \\[-0.05em] \text{a job has rank~$z$ or less under SRPT-E} \end{tabular}}}^2}}{(1 - \rho_Z(z))^2}\\
            &\ge
        \frac{\lambda}{2}
        \E[\Big]{\gp[\Big]{{\footnotesize \begin{tabular}{@{}l@{}} \text{amount of service time during which} \\[-0.05em] \text{a job has rank~0 or less under SRPT-E} \end{tabular}}}^2}.
    \end{align*}
    Note that this lower bound applies regardless of $z$.

    Because $Z = \beta S$,
    a job of size $s$ starts at rank $\beta s$,
    and reaches rank 0 at age $\beta s$.
    Therefore, it receives $(1-\beta)s$ service with rank $\le 0$.
    Therefore, we can lower bound waiting time explicitly:
    \begin{align*}
        \E{\Twait_\srpte} \ge \frac{\lambda}{2} \E[\big]{\gp[\big]{(1-\beta)S}^2}.
    \end{align*}
    From this result, we see that SRPT-E's response time grows with $\E{S^2}$, which can be arbitrarily large or infinite. In contrast, the response time of SRPT is finite even for job size distributions with infinite $\E{S^2}$,
    such as a Pareto distribution with exponent $1.5$.
\end{proof}

\begin{remark}
    While the proof of \cref{thm:srpt-e_result}
    assumes $Z = \beta S$ for simplicity,
    essentially the same result holds
    for any error distribution $Z$
    that is multiplicatively smaller than $S$ with positive probability.
    In particular, suppose that there exists $\beta' < 1$ such that for all~$s$,
    we have $\P{Z \le \beta' S \given S = s} \geq p$.
    Then by essentially the same argument as \cref{thm:srpt-e_result},
    \begin{align*}
        \E{\Twait_\srpte}
        \geq \frac{\lambda p}{2} \E[\big]{\gp[\big]{(1-\beta')S}^2}.
    \end{align*}
    In particular, this means it is possible for SRPT-E
    to have unbounded approximation ratio even when size estimates are unbiased,
    namely when $\E{Z \given S = s} = s$ for all~$s$.
\end{remark}

\section{SRPT with Scaling Estimates (SRPT-SE)}
\label{sec:srpt-se}

SRPT-SE is a policy that uses a job's true size and estimated size to assign its rank.
SRPT-SE is thus not a practical policy,
as one would prefer SRPT if true sizes were known.
However, analyzing it is helpful for a few reasons.
First, it is a useful warmup for the analysis of PSJF-E,
which follows the same outline but is somewhat more complicated
(see \cref{sec:psjf-e}).
Second, it is the first step of analyzing SRPT-B,
whose performance we bound relative to SRPT-SE
(see \cref{sec:srpt-b}).
Third, there are settings in which a policy similar to SRPT-SE,
which enjoys similarly good performance,
could be implemented in practice
(see \cref{sec:continuous-estimate}).

Our main tool for analyzing SRPT-SE is \cref{thm:response_time_work_integral},
which expresses mean response time in terms of mean $(\remsize \leq r)$-work,
with less $(\remsize \leq r)$-work corresponding to lower response time.

Before analyzing SRPT-SE, it is helpful to consider
how one might use \cref{thm:response_time_work_integral} to show that
SRPT minimizes mean response time.
The key is that for every value of~$r$,
SRPT minimizes mean $(\remsize \leq r)$-work,
or equivalently mean $(\rank{\srpt} \leq r)$-work.
The intuition is that
whenever the system has nonzero $(\rank{\srpt} \leq r)$-work,
SRPT serves a job of rank $r$ or less,
thus reducing the amount of $(\rank{\srpt} \leq r)$-work.

It turns out that an analogous property holds
for any policy that can be defined using a rank function \cite{scully_soap_2018},
including those in \cref{def:policy}.

\begin{proposition}[\textnormal{%
    very similar to \cite[Theorem~VII.7]{scully_gittins_2021}\protect\footnotemark}]
    \label{thm:soap_match_minimizes_work}
    Consider a policy $\generic \in \policiesWithTrue$
    and any rank~$r$.
    In the M/G/1, the policy that minimizes
    the mean amount of steady-state $(\rank{\generic} \leq r)$-work
    is $\generic$ itself.
    That is, for any policy~$\genericOther$,
    \begin{equation*}
        \E{W_{\generic}(\rank{\generic} \leq r)}
        \leq \E{W_{\genericOther}(\rank{\generic} \leq r)}.
    \end{equation*}
\end{proposition}

\footnotetext{%
    While \citet[Theorem~VII.7]{scully_gittins_2021}
    consider a specific policy, namely a generalization of SRPT,
    the same proof applies virtually verbatim
    to any policy that can be defined by a rank function
    \cite{scully_soap_2018}.}

Because $\remsizee = \rank{\srptse}$,
we have from \cref{thm:soap_match_minimizes_work} that
SRPT-SE minimizes mean $(\remsizee \leq r)$-work.
But \cref{thm:response_time_work_integral} gives mean response time
in terms of mean $(\remsize \leq r)$-work, not mean $(\remsizee \leq r)$-work.
Fortunately, we can leverage the fact that we have
$(\beta, \alpha)$-bounded size estimates
to relate $(\remsize \leq r)$-work to $(\remsizee \leq r)$-work,
yielding the following result.

\begin{theorem}[Performance of SRPT-SE]
    \label{thm:srpt-se_result}
    Consider the M/G/1 with $(\beta, \alpha)$-bounded size estimates.
    \begin{enumerate}[(a)]
    \item
        The mean response time of SRPT-SE is bounded above by
        \begin{equation*}
            \E{T_\srptse} \leq \frac{\alpha}{\beta} \E{T_\srpt},
        \end{equation*}
    \item
        The approximation ratio of SRPT-SE is at most~$\alpha/\beta$.
    \item
        As $\alpha$ and~$\beta$ converge to~$1$,
        the approximation ratio of SRPT-SE converges to~$1$
        uniformly in the arrival rate
        and the joint distribution of true and estimated sizes.
    \end{enumerate}
\end{theorem}

\begin{proof}
    It clearly suffices to prove~(a).
    Recall the following facts about job states~$x$:
    \begin{itemize}
        \item $\rank{\srpt}(x) = \remsize(x)$,
        \item $\rank{\srptse}(x) = \remsizee(x)$, and
        \item $\remsizee(x)/\remsize(x) = \sizee(x)/\size(x) \in [\beta, \alpha]$.
    \end{itemize}
    Using the above facts together with
    \cref{thm:response_time_work_integral, thm:soap_match_minimizes_work},
    we compute
    \begin{align*}
        \MoveEqLeft
        \E{T_\srptse} \\
        &= \frac{1}{\lambda} \int_0^\infty \frac{\E{W_\srptse(\remsize \leq r)}}{r^2} \d{r}
        \by{\cref{thm:response_time_work_integral}} \\
        &\leq \frac{1}{\lambda} \int_0^\infty \frac{\E{W_\srptse(\remsizee \leq \alpha r)}}{r^2} \d{r}
        \by{using $\remsizee(x)/\remsize(x) \leq \alpha$} \\
        &\leq \frac{1}{\lambda} \int_0^\infty \frac{\E{W_\srpt(\remsizee \leq \alpha r)}}{r^2} \d{r}
        \by{\cref{thm:soap_match_minimizes_work}} \\
        &\leq \frac{1}{\lambda} \int_0^\infty \frac{\E[\big]{W_\srpt\gp[\big]{\remsize \leq \frac{\alpha}{\beta} r}}}{r^2} \d{r}
        \by{using $\remsizee(x)/\remsize(x) \geq \beta$}\\
        &= \frac{\alpha}{\beta} \frac{1}{\lambda} \int_0^\infty \frac{\E{W_\srpt(\remsize \leq r')}}{(r')^2} \d{r'}
        \by{setting $r' = \alpha / \beta \cdot r$} \\
        &= \frac{\alpha}{\beta} \E{T_\srpt}.
        \by{\cref{thm:response_time_work_integral}}
        \qedhere
    \end{align*}
\end{proof}

\Cref{thm:srpt-se_result} completes our analysis of SRPT-SE.  We describe a related result for a similar scheme that does not use the job's true size in \cref{sec:continuous-estimate}.

\section{SRPT with Bounce (SRPT-B)}
\label{sec:srpt-b}

As we have mentioned, previous works have noted that when using SRPT-E, large jobs that are underestimated will have estimated remaining sizes that become negative while the true remaining time is still relatively large, leading to long waiting times for jobs stuck behind them.  An open question has been to modify SRPT with estimated sizes in a way that avoids this issue in a robust manner, without assumptions on job size distributions.  Our suggested solution, SRPT with Bounce (SRPT-B), handles this by modifying the rank function from $z-a$ to $\min\{|z-a|,z\}$.\footnote{We note that it is an interesting open question to consider the effects of other possible forms for the bounce, which we do not investigate here.}  We prove the following results for SRPT-B.

\restate{thm:srpt-b_result}

Our overall approach to analyzing SRPT-B is to compare it to SRPT-SE.
We separately compare the waiting times and residence times of the two policies
(see \cref{sec:proof_overview:soap}).
Both comparisons boil down to comparing
the amount of time a job spends above or below a given rank under each policy.
We begin with the residence time comparison
(see \cref{sec:srpt-b:residence})
before moving on to the more complicated waiting time comparison
(see \cref{sec:srpt-b:waiting}).
Combining the two comparisons and some additional computation
(see \cref{sec:srpt-b:response})
yields \cref{thm:srpt-b_result}.

\subsection{Residence Time Difference between SRPT-B and SRPT-SE}
\label{sec:srpt-b:residence}

By \cref{thm:soap}(b),
the expected residence time of the job under SRPT-B, SRPT-SE, or PSJF-E is
an integral of $1/(1 - \rho_Z(\cdot))$ terms over the job's ages,
where the value plugged is the worst future rank of the job.

Consider a job of true size~$s$ and estimated size~$z$.
In order to bound $\E{\Tres_\srpt(s, z)}$,
we will find functions $g_{s, z}(\cdot), h_{s, z}(\cdot)$
and a value $c_{s, z} > 0$ such that
\begin{align}
    \label{eq:g_integral}
    \int_0^{(1 + c_{s, z})s} \frac{1}{1 - \rho_Z(g_{s, z}(a))} \d{a} &= \E{\Tres_\srptb(s, z)} + c_{s, z}s, \\
    \label{eq:h_integral}
    \int_0^{(1 + c_{s, z})s} \frac{1}{1 - \rho_Z(h_{s, z}(a))} \d{a} &= \E{\Tres_\srptse(s, z)} + c_{s, z}\E{\Tres_\psjfe(s, z)}, \\
    \label{eq:g_leq_h}
    g_{s, z}(a) &\leq h_{s, z}(a) \quad \text{for all~$a \in (0, (1 + c_{s, z})s)$.}
\end{align}
Because $1/(1 - \rho_Z(\cdot))$ is nondecreasing, this implies
\begin{equation*}
    \E{\Tres_\srptb(s, z)} \leq \E{\Tres_\srptse(s, z)} + c_{s, z}(\E{\Tres_\psjfe(s, z)} - s).
\end{equation*}
Finally, we will bound $c_{s, z}$ by a value~$c$ which is independent of $s$ and~$z$,
obtaining
\begin{equation}
    \label{eq:srpt-b_residence_goal}
    \E{\Tres_\srptb} \leq \E{\Tres_\srptse} + c(\E{\Tres_\psjfe} - \E{S}).
\end{equation}

We begin by computing the worst future ranks of each policy.

\begin{lemma}
    \label{thm:worst_future_rank}
    The worst future ranks of a job of true size~$s$, estimated size~$z$,
    and age~$a$ under SRPT-B, SRPT-SE, and PSJF-E are
    \begin{align*}
        \worst{\srptb}(s, z, a) &= \max\{z - a, \min\{s - z, z\}\}, \\
        \worst{\srptse}(s, z, a) &= \rank{\srptse}(s, z, a) = \frac{z}{s}(s - a), \\
        \worst{\psjfe}(s, z, a) &= \rank{\psjfe}(s, z, a) = z.
    \end{align*}
\end{lemma}

\begin{proof}
    SRPT-SE and PSJF-E have nondecreasing rank as a function of age~$a$,
    so the job's worst future rank is its current rank.
    If $s \leq z$, then the same is true for SRPT-B.
    If instead $z < s \leq 2z$, then under SRPT-B,
    the job's worst future rank is its current rank~$z - a$
    until age $a = s - 2(s - z)$,
    after which the worst future rank is its final rank, namely $s - z$.
    Finally, if $s > 2z$,
    then the job's worst future rank is always its final rank, namely~$z$.
\end{proof}

\begin{lemma}
    \label{thm:g_h_exist}
    The following definitions satisfy
    \cref{eq:g_integral, eq:h_integral, eq:g_leq_h}:
    \begin{align*}
        c_{s, z} &= \min\curlgp[\Big]{\vert[\Big]{1 - \frac{s}{z}}, 1}, \\
        g_{s, z}(a) &= \begin{cases}
            \max\{z - a, \min\{s - z, z\}\} & \text{if } a \leq s \\
            0 & \text{if } a > s,
        \end{cases} \\
        h_{s, z}(a) &= \begin{cases}
            z & \text{if } a \leq c_{s, z}s \\
            \frac{z}{s}(s - (a - c_{s, z}s)) & \text{if } a > c_{s, z}s.
        \end{cases}
    \end{align*}
\end{lemma}

\premidfigure
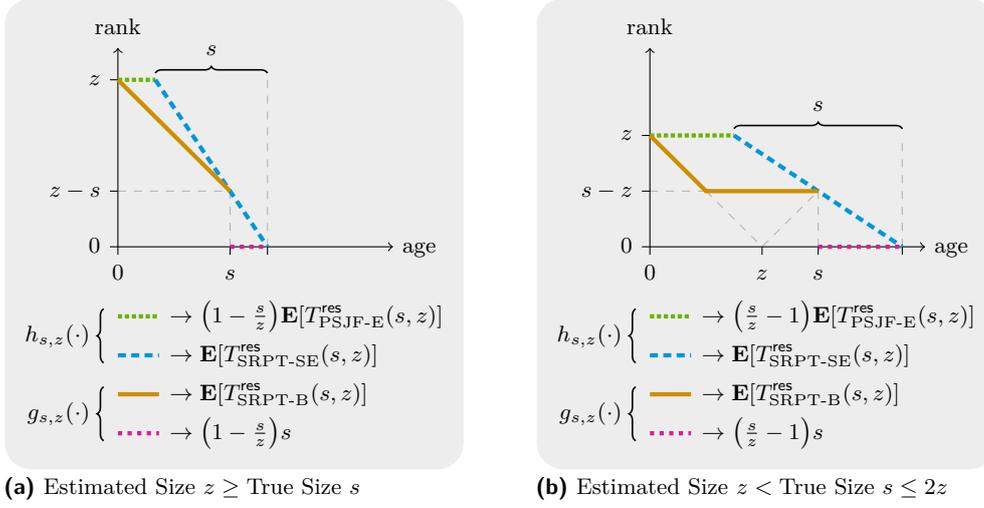
\begin{figure}
    \begin{subfigure}[t]{0.5\linewidth}
        \begin{tikzpicture}[figure]
    \newcommand{\z}{9}
    \newcommand{\s}{6}
    \newcommand{\diffexpr}{(\z-\s)*\s/\z}
    \newcommand{\diff}{{\diffexpr}}
    
    \yguide[$z$]{0}{\z}
    \yguide[$z - s$]{\s}{\z-\s}
    \xguide[$s$]{\s}{\z-\s}
    \xguide[\relax]{{\s+\diffexpr}}{\z}
    
    \draw[decorate, decoration={brace, raise=0.25em}]
      (\diff, \z) -- ++(\s, 0)
      node[midway, above=0.4em]
      {\vphantom{by}$s$};
    
    \axes{14}{10}{$0$}{age}{$0$}{rank}
    \draw[psjf-e, densely dotted] (0, \z) -- ++(\diff, 0);
    \draw[srpt-se, densely dashed] (\diff, \z) -- ++(\s, -\z);
    \draw[srpt-b] (0, \z) -- ++(\s, -\s);
    \draw[size, dotted] (\s, 0) -- ++(\diff, 0);
    
    \begin{scope}[shift={(0, -3.75)}]
        \newcommand{\spacing}{2.1}
        
        \draw[psjf-e, densely dotted] (0, 0) -- ++(2.2, 0);
        \node[right] at (1.9, 0) {\vphantom{by}${}\to \gp[\big]{1 - \frac{s}{z}} \E{\Tres_\psjfe(s, z)}$};
        \draw[srpt-se, densely dashed] (0, -\spacing) -- ++(2.2, 0);
        \node[right] at (1.9, -\spacing) {\vphantom{by}${}\to \E{\Tres_\srptse(s, z)}$};
        \draw[srpt-b] (0, -2*\spacing) -- ++(2.2, 0);
        \node[right] at (1.9, -2*\spacing) {\vphantom{by}${}\to \E{\Tres_\srptb(s, z)}$};
        \draw[size, dotted] (0, -3*\spacing) -- ++(2.2, 0);
        \node[right] at (1.9, -3*\spacing) {\vphantom{by}${}\to \gp[\big]{1 - \frac{s}{z}} s$};
        
        \draw[decorate, decoration={brace, mirror, raise=0.25em}]
            (-0.4, 0 + 0.4) -- ++(0, -\spacing - 2*0.4)
            node[midway, left=0.4em] {\vphantom{by}$h_{s, z}(\cdot)$};
        \draw[decorate, decoration={brace, mirror, raise=0.25em}]
            (-0.4, -2*\spacing + 0.4) -- ++(0, -\spacing - 2*0.4)
            node[midway, left=0.4em] {\vphantom{by}$g_{s, z}(\cdot)$};
    \end{scope}
\end{tikzpicture}
        \caption{Estimated Size~$z \geq {}$True Size~$s$}
    \end{subfigure}\hfill\begin{subfigure}[t]{0.5\linewidth}
        \begin{tikzpicture}[figure]
    \newcommand{\z}{6}
    \newcommand{\s}{9}
    \newcommand{\diffexpr}{(\s-\z)*\s/\z}
    \newcommand{\diff}{{\diffexpr}}
    
    \yguide[$z$]{0}{\z}
    \xguide[$z$]{\z}{0}
    \xguide[$s$]{\s}{\s-\z}
    \yguide[$s - z$]{{\z-(\s-\z)}}{\s-\z}
    \xguide[\relax]{{\s+\diffexpr}}{\z}
    
    \draw[guide] (0, \z) ++({\s-2*(\s-\z)}, {-(\s-2*(\s-\z))}) -- ++({\s-\z}, {-(\s-\z)}) -- ++ ({\s-\z}, {\s-\z});
    
    \draw[decorate, decoration={brace, raise=0.25em}]
      (\diff, \z) -- ++(\s, 0)
      node[midway, above=0.4em]
      {\vphantom{by}$s$};
    
    \axes{14}{10}{$0$}{age}{$0$}{rank}
    \draw[psjf-e, densely dotted] (0, \z) -- ++(\diff, 0);
    \draw[srpt-se, densely dashed] (\diff, \z) -- ++(\s, -\z);
    \draw[srpt-b] (0, \z) -- ++({\s-2*(\s-\z)}, {-(\s-2*(\s-\z))}) -- ++({2*(\s-\z)}, 0);
    \draw[size, dotted] (\s, 0) -- ++(\diff, 0);
    
    \begin{scope}[shift={(0, -3.75)}]
        \newcommand{\spacing}{2.1}
        \draw[psjf-e, densely dotted] (0, 0) -- ++(2.2, 0);
        \node[right] at (1.9, 0) {\vphantom{by}${}\to \gp[\big]{\frac{s}{z} - 1} \E{\Tres_\psjfe(s, z)}$};
        \draw[srpt-se, densely dashed] (0, -\spacing) -- ++(2.2, 0);
        \node[right] at (1.9, -\spacing) {\vphantom{by}${}\to \E{\Tres_\srptse(s, z)}$};
        \draw[srpt-b] (0, -2*\spacing) -- ++(2.2, 0);
        \node[right] at (1.9, -2*\spacing) {\vphantom{by}${}\to \E{\Tres_\srptb(s, z)}$};
        \draw[size, dotted] (0, -3*\spacing) -- ++(2.2, 0);
        \node[right] at (1.9, -3*\spacing) {\vphantom{by}${}\to \gp[\big]{\frac{s}{z} - 1} s$};
        
        \draw[decorate, decoration={brace, mirror, raise=0.25em}]
            (-0.4, 0 + 0.4) -- ++(0, -\spacing - 2*0.4)
            node[midway, left=0.4em] {\vphantom{by}$h_{s, z}(\cdot)$};
        \draw[decorate, decoration={brace, mirror, raise=0.25em}]
            (-0.4, -2*\spacing + 0.4) -- ++(0, -\spacing - 2*0.4)
            node[midway, left=0.4em] {\vphantom{by}$g_{s, z}(\cdot)$};
    \end{scope}
\end{tikzpicture}
        \caption{Estimated Size~$z < {}$True Size~$s \leq 2z$}
    \end{subfigure}
    \caption{Relating Residence Times of SRPT-B, SRPT-SE, and PSJF-E in Proof of \cref{thm:g_h_exist}}
    \label{fig:residence}
\end{figure}
\postmidfigure

\begin{proof}
    By \cref{thm:worst_future_rank}, we can view
    $g_{s, z}(\cdot)$ as
    SRPT-B's worst future rank followed by zero.
    Similarly, $h_{s, z}(\cdot)$ is
    a segment of PSJF-E's worst future rank
    followed by SRPT-SE's worst future rank (shifted to the right).
    Applying \cref{thm:soap}(b)
    thus yields \cref{eq:g_integral, eq:h_integral}.

    It remains only to prove \cref{eq:g_leq_h},
    namely that $g_{s, z}(\cdot)$ is always below $h_{s, z}(\cdot)$.
    We show this geometrically in \cref{fig:residence}.
    The illustration shows the $z \geq s$ and $z < s \leq 2z$ cases,
    and the $s > 2z$ case is essentially the same as the $s = 2z$ case.
\end{proof}

\begin{proposition}
    \label{thm:srpt-b_residence}
    The mean residence time of SRPT-B is bounded above by
    \begin{equation*}
        \E{\Tres_\srptb} \leq \E{\Tres_\srptse} + \min\curlgp[\bigg]{1, \max\curlgp[\Big]{1 - \frac{1}{\alpha}, \frac{1}{\beta} - 1}} \gp*{\frac{1}{\rho} \ln\frac{1}{1 - \rho} - 1} \E{S}.
    \end{equation*}
\end{proposition}

\begin{proof}
    Let $c_{s, z} = \min\{1, |1 - s/z|\}$, as in \cref{thm:g_h_exist}.
    Because we have $(\beta, \alpha)$-bounded size estimates,
    $c_{s, z} \leq \min\{1, \max\{1 - 1/\alpha, 1/\beta - 1\}\}$
    for all feasible pairs of true size~$s$ and estimated size~$z$.
    This bound on $c_{s, z}$ is independent of $s$ and~$z$,
    so by \cref{thm:g_h_exist} and the discussion at the start of this section,
    \cref{eq:srpt-b_residence_goal} holds
    with $c = \min\{1, \max\{1 - 1/\alpha, 1/\beta - 1\}\}$.
    The result then follows from \cref{thm:psjf-e_residence},
    which gives the value of~$\E{\Tres_\psjfe}$.
\end{proof}

\subsection{Waiting Time Difference between SRPT-B and SRPT-SE}
\label{sec:srpt-b:waiting}

By \cref{thm:soap}(a),
computing the waiting time of SRPT-B and SRPT-SE
boils down to computing how much of a job's service happens below a given rank.
That is, recalling
\begin{equation*}
    u_\generic(z) = \E[\bigg]{\gp[\bigg]{\begin{tabular}{@{}l@{}}\text{amount of service time during which} \\ \text{a job has rank $z$ or less under policy~$\generic$}\end{tabular}}^2},
\end{equation*}
our goal is to compare $u_\srptb(z)$ and~$u_\srptse(z)$.
We begin by computing both quantities.

\begin{lemma}
    \label{thm:u_exact}
    \begin{align*}
        u_\srptb(z) &= \E[\big]{S^2 \1(Z \leq z) + \gp[\big]{\max\{0, \min\{S - (Z - z), 2z\}\}}^2 \1(Z > z)}, \\
        u_\srptse(z) &= \E[\bigg]{S^2 \1(Z \leq z) + \gp[\bigg]{\frac{S}{Z}z}^2 \1(Z > z)}.
    \end{align*}
\end{lemma}

\begin{proof}
    Consider a job with true size~$S$ and estimated size~$Z$
    drawn from the joint distribution of true and estimated sizes.
    Under both policies, if $Z \leq z$,
    then the job's rank remains $z$ or less for its entire service time,
    which explains the $\1(Z \leq z)$ terms.
    If instead $Z > z$, then the following reasoning explains the $\1(Z > z)$ terms.
    \begin{itemize}
        \item Under SRPT-B,
        the job has rank $z$ or less when its age is in interval $[0, S) \cap [Z - z, Z + z]$.
        \item Under SRPT-SE,
        the job spends a $z/Z$ fraction of its service time with rank~$z$ or less.
        \qedhere
    \end{itemize}
\end{proof}

\begin{lemma}
    \label{thm:u_bound}
    \begin{equation*}
        u_\srptb(z) - u_\srptse(z) \leq 3z \max\{1 - \beta, 0\} \E{S \1(Z > z)}.
    \end{equation*}
\end{lemma}

\begin{proof}
    We begin by applying \cref{thm:u_exact}:
    \begin{equation*}
        u_\srptb(z) - u_\srptse(z)
        = \E[\bigg]{\gp[\bigg]{\gp[\big]{\max\{0, \min\{S - (Z - z), 2z\}\}}^2 - \gp[\bigg]{\frac{S}{Z}z}^2} \1(Z > z)}.
    \end{equation*}
    If $S \leq Z$, then because $z/Z < 1$ whenever the indicator is nonzero,
    the following computation shows that
    the expression inside the expectation is nonpositive:
    \begin{equation*}
        \begin{aligned}
            \frac{z}{Z} &< 1 &
            &{\Rightarrow} & \frac{z}{Z}(Z - S) &\leq Z - S &
            &{\Rightarrow} & S - (Z - z) &\leq \frac{S}{Z}z.
        \end{aligned}
    \end{equation*}
    This means adding an $S > Z$ to the indicator gives an upper bound,
    from which we compute
    \begin{align*}
        \MoveEqLeft
        u_\srptb(z) - u_\srptse(z) \\
        &\leq \E[\Bigg]{\gp[\Bigg]{\gp[\big]{\max\{0, \min\{S - (Z - z), 2z\}\}}^2 - \gp[\bigg]{\frac{S}{Z}z}^2} \1(S > Z > z)} \\
        & \by[-96mu]{nonpositive when $S \leq Z$} \\
        &\leq \E[\big]{\gp[\big]{(z + \min\{S - Z, z\})^2 - z^2}\1(S > Z > z)}  \\
        &\leq \E{3z(S - Z) \1(S > Z > z)} \\
        &\leq 3z \max\{1 - \beta, 0\} \E{S \1(Z > z)}.
        \by[-96mu]{using $Z \geq \beta S$}
        \qedhere
    \end{align*}
\end{proof}

\begin{proposition}
    \label{thm:srpt-b_waiting}
    The mean waiting time of SRPT-B is bounded above by
    \begin{equation*}
        \E{\Twait_\srptb}
        \leq \E{\Twait_\srptse} + \frac{3}{2} \alpha \max\{1 - \beta, 0\} \gp*{\frac{1}{\rho} \ln\frac{1}{1 - \rho} - 1} \E{S}.
    \end{equation*}
\end{proposition}

\begin{proof}
    The high-level steps of the proof are the following:
    \begin{itemize}
        \item We use \cref{thm:soap}(a)
        to express the difference $\E{\Twait_\srptb} - \E{\Twait_\srptse}$
        in terms of $u_\srptb(z) - u_\srptse(z)$.
        \item We bound $u_\srptb(z) - u_\srptse(z)$ using \cref{thm:u_bound}.
        \item We use integration by parts,
        obtaining an expression that is similar to one that appears in
        the proof of \cref{thm:psjf-e_residence}.
        \item Mirroring the remainder of the proof of \cref{thm:psjf-e_residence},
        which involves applying \cref{thm:rho_derivative},
        yields the desired result.
    \end{itemize}
    We begin by computing
    \begin{align*}
        \MoveEqLeft
        \E{\Twait_\srptb} - \E{\Twait_\srptse} \\
        &= \int_0^\infty (\E{\Twait_\srptb(z)} - \E{\Twait_\srptse(z)}) f_Z(z) \d{z}
        \by{conditioning on $Z$} \\
        &= \frac{\lambda}{2} \int_0^\infty \frac{u_\srptb(z) - u_\srptse(z)}{(1 - \rho_Z(z))^2} f_Z(z) \d{z}
        \by{\cref{thm:soap}(a)} \\
        &\leq \frac{3}{2} \max\{1 - \beta, 0\} \int_0^\infty \frac{\lambda z \E{S \1(Z > z)}}{(1 - \rho_Z(z))^2} f_Z(z) \d{z}
        \by{\cref{thm:u_bound}} \\
        &\leq \frac{3}{2} \alpha \max\{1 - \beta, 0\} \int_0^\infty \frac{\lambda \E{S \given Z = z} \E{S \1(Z > z)}}{(1 - \rho_Z(z))^2} f_Z(z) \d{z}.
        \by{using $Z \leq \alpha S$}
    \end{align*}
    It remains only to bound the last integral.
    By \cref{thm:rho_derivative}, we have
    \begin{equation}
        \label{eq:waiting_integration_by_parts_1}
        \frac{\lambda \E{S \given Z = z} f_Z(z)}{(1 - \rho_Z(z))^2}
        = \dd{z}\frac{1}{1 - \rho},
    \end{equation}
    and conditioning on~$Z$ yields
    \begin{equation}
        \label{eq:waiting_integration_by_parts_2}
        \E{S \1(Z > z)} = \int_z^\infty \E{S \given Z = z'} \d{z'}.
    \end{equation}
    Combining these equations and integrating by parts, we obtain
    \begin{align*}
        \MoveEqLeft
        \int_0^\infty \frac{\lambda \E{S \given Z = z} \E{S \1(Z > z)}}{(1 - \rho_Z(z))^2} f_Z(z) \d{z} \\
        &= \int_0^\infty \gp[\bigg]{\dd{z}\frac{1}{1 - \rho_Z(z)}} \gp[\bigg]{\int_z^\infty \E{S \given Z = z'} f_Z(z') \d{z'}} \d{z}
        \by{\cref{eq:waiting_integration_by_parts_1, eq:waiting_integration_by_parts_2}} \\
        &= 0 - \E{S} + \int_0^\infty \frac{\E{S \given Z = z} f_Z(z)}{1 - \rho_Z(z)} \d{z}
        \by{integrating by parts} \\
        &= \gp*{\frac{1}{\rho} \ln\frac{1}{1 - \rho} - 1} \E{S}.
        \by{\cref{thm:rho_derivative}}
        \qedhere
    \end{align*}
\end{proof}

\subsection{Combining Waiting Time and Residence Time Bounds}
\label{sec:srpt-b:response}

\begin{proof}[Proof of \cref{thm:srpt-b_result}]
    By \cref{thm:srpt_log}, (a) implies (b) and~(c).
    It suffices to show~(a),
    which follows by combining \cref{thm:srpt-b_residence, thm:srpt-b_waiting},
    applying \cref{thm:srpt-se_result}(a),
    and observing
    \begin{equation*}
        \max\{1 - \beta, 0\}
        \leq \1(\beta < 1) \min\curlgp[\bigg]{1, \max\curlgp[\Big]{1 - \frac{1}{\alpha}, \frac{1}{\beta} - 1}}.
        \qedhere
    \end{equation*}
\end{proof}

\section{PSJF with Estimates (PSJF-E)}
\label{sec:psjf-e}

\restate{thm:psjf-e_result}

We prove \cref{thm:psjf-e_result},
which compares PSJF and PSJF-E,
using an argument similar to the proof of \cref{thm:srpt-se_result},
which compares SRPT and SRPT-SE.
However, because PSJF prioritizes jobs by \emph{original} size,
as opposed to \emph{remaining} size,
combining \cref{thm:response_time_work_integral, thm:soap_match_minimizes_work}
to compare PSJF with PSJF-E is not as straightforward
as comparing SRPT with SRPT-SE.
We begin by working out how \cref{thm:response_time_work_integral}
applies to PSJF and PSJF-E.

\begin{lemma}
    \label{thm:psjf_work_integral}
    The mean response time of PSJF is
    \begin{equation*}
        \E{T_\psjf}
        = \frac{1}{\lambda} \int_0^\infty \frac{\E{W_\psjf(\size \leq r)}}{r^2} \d{r}
            + \frac{1}{2\lambda} \ln\frac{1}{1-\rho}.
    \end{equation*}
\end{lemma}

\begin{proof}
    Applying \cref{thm:response_time_work_integral} yields
    \begin{align*}
        \E{T_\psjf}
        &= \frac{1}{\lambda} \int_0^\infty \frac{\E{W_\psjf(\remsize \leq r)}}{r^2} \d{r} \\
        \yestag
        \label{eq:T_psjf_first_step}
        &= \frac{1}{\lambda} \int_0^\infty \frac{\E{W_\psjf(\size \leq r)}}{r^2} \d{r}
            + \frac{1}{\lambda} \int_0^\infty \frac{\E{W_\psjf(\remsize \leq r < \size)}}{r^2} \d{r}.
    \end{align*}
    It remains only to compute $\E{W_\psjf(\remsize \leq r < \size)}$.

    Let $\phi_r(s, z, a) = (s - a \leq r < s)$.
    That is, $\phi_r$ is true for states with remaining size less than $r$
    but original size greater than~$r$.
    Our goal is to compute $\E{W(\phi)}$.
    Recall from \cref{def:phi-work} that $W(\phi_r)$ is
    the sum of each individual job's $\phi_r$-work.
    We can therefore use a generalization of Little's law
    \cite{brumelle_relation_1971, heyman_relation_1980}
    to express the average total amount of $\phi_r$-work, namely $\E{W(\phi_r)}$,
    as the arrival rate~$\lambda$ times
    the average cumulative amount of $\phi_r$-work a job contributes
    over the course of its time in the system.
    Specifically, consider a random job with
    true size~$S$, estimated size~$Z$, response time~$T$,
    and age~$A(t)$ after being in the system for time $t \in [0, T)$.\footnote{%
        The random variables $S$, $Z$, $T$, and $A(t)$ refer to the same job,
        so they are not independent.
        In addition, the response time~$T$ and age~$A(t)$ depend on PSJF scheduling.
        The same applies to $\Twait$ and~$\Tres$, which we introduce shortly.}
    Then by the generalization of Little's law, we can write the mean $\phi_r$-work as
    \begin{align}
        \label{eq:little_w_phi_r}
        \E{W_\psjf(\phi_r)} = \lambda \E[\bigg]{\int_0^T w(\phi_r, (S, Z, A(t))) \d{t}},
    \end{align}

    We now compute the expectation of the right-hand side of \cref{eq:little_w_phi_r},
    which concerns a single job's time in the system.
    Because $\phi_r(s, z, a)$ holds only if $a > 0$
    and a job's age is~$0$ during its waiting time,
    we can restrict attention to residence time.
    Letting $\Twait$ and $\Tres$ denote the random job's
    waiting and residence times, respectively,
    we have
    \begin{align*}
        \MoveEqLeft
        \E[\bigg]{\int_0^T w(\phi_r, (S, Z, A(t))) \d{t}} \\
        &= \E[\bigg]{\int_0^{\Twait} w(\phi_r, (S, Z, 0)) \d{t} + \int_{\Twait}^{\Twait + \Tres} w(\phi_r, (S, Z, A(t))) \d{t}} \\
        \yestag
        \label{eq:response_integral_to_residence_integral}
        &= \E[\bigg]{\int_{\Twait}^{\Twait + \Tres} w(\phi_r, (S, Z, A(t))) \d{t}}.
    \end{align*}
    Under PSJF, a job's rank is always its size,
    and the load of jobs with rank better than~$s$ is~$\rho_S(s)$.
    This means that during the residence time of a job of size~$s$,
    it takes $\Delta/(1 - \rho_S(s))$ time for a job's age to increase by~$\Delta$.\footnote{%
        This $1/(1 - \rho_S(s))$ factor under PSJF
        is similar to the $1/(1 - \rho_Z(z))$ factor that appears under PSJF-E,
        as discussed in \cref{rmk:residence_busy_period}.
        One can use the concept of \emph{busy periods}
        from M/G/1 scheduling theory to formalize this
        \cite{harchol-balter_performance_2013, scully_soap_2018}.}
    This means that for any function~$g$ and any size~$s$,
    \begin{equation}
        \label{eq:residence_integral_to_age_integral}
        \E[\bigg]{\int_{\Twait}^{\Twait + \Tres} g(A(t)) \d{t} \given S = s} = \frac{1}{1 - \rho_S(s)} \int_0^s g(a) \d{a}.
    \end{equation}
    From this, we compute
    \begin{align*}
        \MoveEqLeft
        \E[\bigg]{\int_0^T w(\phi_r, (S, Z, A(t))) \d{t}} \\
        &= \E[\bigg]{\int_{\Twait}^{\Twait + \Tres} w(\phi_r, (S, Z, A(t))) \d{t}}
        \by[-72mu]{\cref{eq:response_integral_to_residence_integral}} \\
        &= \int_0^\infty \E[\bigg]{\int_{\Twait}^{\Twait + \Tres} w(\phi_r, (s, Z, A(t))) \d{t} \given S = s} f_S(s) \d{s}
        \by[-72mu]{conditioning on $S$} \\
        &= \int_0^\infty \E[\bigg]{\int_{\Twait}^{\Twait + \Tres} (s - A(t))\1(s - A(t) \leq r < s) \d{t} \given S = s} f_S(s) \d{s} \\
        & \by[-72mu]{expanding $\phi_r$} \\
        &= \int_0^\infty \nestint_0^s (s - a) \1(s - a \leq r < s) \d{a} \, \frac{f_S(s)}{1 - \rho_S(s)} \d{s}
        \by[-72mu]{\cref{eq:residence_integral_to_age_integral}} \\
        \yestag
        \label{eq:cumulative_job_rrs-work}
        &= \frac{r^2}{2} \int_r^\infty \frac{f_S(s)}{1 - \rho_S(s)} \d{s}.
    \end{align*}
    Finally, we use \cref{eq:little_w_phi_r}
    to plug this back into \cref{eq:T_psjf_first_step},
    obtaining
    \begin{align*}
        \frac{1}{\lambda} \int_0^\infty \frac{\E{W_\psjf(\remsize \leq r < \size)}}{r^2} \d{r}
        &= \frac{1}{2} \int_0^\infty \nestint_r^\infty \frac{f_S(s)}{1 - \rho_S(s)} \d{s} \d{r}
        \by[-18mu]{\cref{eq:little_w_phi_r, eq:cumulative_job_rrs-work}} \\
        &= \frac{1}{2} \int_0^\infty \frac{s f_S(s)}{1 - \rho_S(s)} \d{s}
        \by[-18mu]{swapping integrals} \mkern-48mu \\
        &= \frac{1}{2\lambda} \ln\frac{1}{1 - \rho}.
        \by[-18mu]{\cref{thm:rho_derivative}}
        \qedhere
    \end{align*}
\end{proof}

\begin{lemma}
    \label{thm:psjf-e_work_integral}
    The mean response time of PSJF-E is bounded by
    \begin{equation*}
        \E{T_\psjfe}
        \leq \frac{1}{\lambda} \int_0^\infty \frac{\E{W_\psjfe(\sizee \leq \alpha r)}}{r^2} \d{r}
            + \frac{\alpha}{\beta} \frac{1}{2\lambda} \ln\frac{1}{1-\rho}.
    \end{equation*}
\end{lemma}

\begin{proof}
    Applying \cref{thm:response_time_work_integral}
    and using the fact that $\remsizee(x)/\remsize(x) \leq \alpha$ yields
    \begin{align*}
        \E{T_\psjfe}
        &= \frac{1}{\lambda} \int_0^\infty \frac{\E{W_\psjfe(\remsize \leq r)}}{r^2} \d{r} \\
        &\leq \frac{1}{\lambda} \int_0^\infty \frac{\E{W_\psjfe(\remsizee \leq \alpha r)}}{r^2} \d{r} \\
        &= \frac{1}{\lambda} \int_0^\infty \frac{\E{W_\psjfe(\sizee \leq \alpha r)}}{r^2} \d{r} \\
        \yestag
        \label{eq:T_psjf-e_first_step}
            &\quad + \frac{1}{\lambda} \int_0^\infty \frac{\E{W_\psjfe(\remsizee \leq r < \sizee)}}{r^2} \d{r}.
    \end{align*}
    It remains only to bound $\E{W_\psjfe(\remsizee \leq \alpha r < \sizee)}$.

    The first part of this computation
    is similar to part of the proof of \cref{thm:psjf_work_integral}.
    Specifically, under PSJF-E,
    a job's rank is always its estimated size,
    so analogues of
    \cref{eq:response_integral_to_residence_integral, eq:residence_integral_to_age_integral, eq:cumulative_job_rrs-work}
    hold for PSJF-E.
    The main difference is that in \cref{eq:residence_integral_to_age_integral},
    we condition on a job having both size~$s$ and estimated size~$z$,
    and we use $\rho_Z(z)$ instead of~$\rho_S(s)$.
    The end result is
    \begin{align*}
        \MoveEqLeft
        \frac{1}{\lambda}\E{W_\psjfe(\remsizee \leq \alpha r < \sizee)} \\
        &= \int_0^\infty \nestint_0^\infty \nestint_0^s (s - a) \1\gp*{\frac{z}{s}(s - a) \leq \alpha r < z} \d{a} \, \frac{f_{S, Z}(s, z)}{1 - \rho_Z(z)} \d{s} \d{z}.
    \end{align*}
    Simplifying the right-hand side yields
    \begin{align*}
        \MoveEqLeft
        \frac{1}{\lambda}\E{W_\psjfe(\remsizee \leq \alpha r < \sizee)} \\
        &= \int_0^\infty \nestint_0^\infty \nestint_0^s q \1\gp*{\frac{z}{s} q \leq \alpha r < z} \d{q} \, \frac{f_{S, Z}(s, z)}{1 - \rho_Z(z)} \d{s} \d{z}
        \by{setting $q = s - a$} \\
        &= \int_{\alpha r}^\infty \nestint_0^\infty \frac{\frac{1}{2} \gp[\big]{\frac{\alpha r s}{z}}^2}{1 - \rho_Z(z)} \d{q} \, f_{S, Z}(s, z) \d{s} \d{z} \\
        &= \int_{\alpha r}^\infty \frac{\frac{1}{2} \E[\big]{\gp[\big]{\frac{\alpha r S}{z}}^2 \given Z = z} f_Z(z)}{1 - \rho_Z(z)} \d{z} \\
        \yestag
        \label{eq:cumulative_job_yrz-work}
        &= \alpha \frac{r^2}{2} \int_{\alpha r}^\infty \frac{\frac{\alpha}{z} \E[\big]{\gp[\big]{\frac{S}{z}} S \given Z = z} f_Z(z)}{1 - \rho_Z(z)} \d{z}.
    \end{align*}
    Finally, we plug this back into \cref{eq:T_psjf-e_first_step}, obtaining
    \begin{align*}
        \MoveEqLeft
        \frac{1}{\lambda} \int_0^\infty \frac{\E{W_\psjfe(\remsizee \leq \alpha r < \sizee)}}{r^2} \d{r} \\
        &= \alpha \frac{1}{2} \int_0^\infty \nestint_{\alpha r}^\infty \frac{\frac{\alpha}{z} \E[\big]{\gp[\big]{\frac{S}{z}} S \given Z = z} f_Z(z)}{1 - \rho_Z(z)} \d{z} \d{r}
        \by{\cref{eq:cumulative_job_yrz-work}} \\
        &= \alpha \frac{1}{2} \int_0^\infty \frac{\E[\big]{\gp[\big]{\frac{S}{z}} S \given Z = z} f_Z(z)}{1 - \rho_Z(z)} \d{z}
        \by{swapping integrals} \\
        &\leq \frac{\alpha}{\beta} \frac{1}{2} \int_0^\infty \frac{\E{S \given Z = z} f_Z(z)}{1 - \rho_Z(z)} \d{z}
        \by{using $Z/S \geq \beta$} \\
        &= \frac{\alpha}{\beta} \frac{1}{2\lambda} \ln\frac{1}{1 - \rho}.
        \by{\cref{thm:rho_derivative}}
        \qedhere
    \end{align*}
\end{proof}

\begin{proof}[Proof of \cref{thm:psjf-e_result}]
    \Citet[Theorem~5.1]{wierman_nearly_2005} show that
    $\E{T_\psjf} \leq 1.5 \E{T_\srpt}$,
    which means (a) implies~(b),
    so it suffices to prove~(a).
    \Cref{thm:psjf_work_integral}
    expresses $\E{T_\psjf}$ as a sum of two terms,
    and \cref{thm:psjf-e_work_integral} bounds $\E{T_\psjfe}$ by
    a sum of two similar terms.
    The ratio of the second terms is~$\alpha/\beta$.
    To bound the ratio of the first terms by~$\alpha/\beta$,
    we proceed similarly to the proof of \cref{thm:srpt-se_result}:
    \begin{align*}
        \MoveEqLeft
        \frac{1}{\lambda} \int_0^\infty \frac{\E{W_\psjfe(\sizee \leq \alpha r)}}{r^2} \d{r} \\
        &\leq \frac{1}{\lambda} \int_0^\infty \frac{\E{W_\psjf(\sizee \leq \alpha r)}}{r^2} \d{r}
        \by{\cref{thm:soap_match_minimizes_work}} \\
        &\leq \frac{1}{\lambda} \int_0^\infty \frac{\E[\big]{W_\psjf\gp[\big]{\size \leq \frac{\alpha}{\beta} r}}}{r^2} \d{r}
        \by{using $\sizee(x)/\size(x) \geq \beta$} \\
        &= \frac{\alpha}{\beta} \frac{1}{\lambda} \int_0^\infty \frac{\E{W_\psjf(\size \leq r')}}{(r')^2} \d{r'}.
        \by{setting $r' = \alpha/\beta \cdot r$}
        \qedhere
    \end{align*}
\end{proof}

\section{Conclusion}

We have examined scheduling policies in the context of estimated sizes.  Our main result is to resolve an issue previously seen empirically, namely that the performance of SRPT-E can degrade significantly due to long jobs being underestimated, by developing and analyzing a novel policy, SRPT-B, which combines the best aspects of SRPT-E and PSJF-E.  In analyzing SRPT-B, we have demonstrated that it has three key properties for $(\beta,\alpha)$-bounded estimates:
(a) an approximation ratio near 1 when $\alpha$ and $\beta$ are near~1, (b) an approximation ratio bounded by some function of $\alpha$ and~$\beta$, and (c) implementation without knowledge of $\alpha$ and~$\beta$.  We have also shown that PSJF-E also has properties (b) and (c), and has an approximation ratio near 1 relative to PSJF when $\alpha$ and $\beta$ are near~1. We have also shown that the empirical observation of the poor performance of SRPT-E can be characterized through a lower bound.

For practical settings, our results provide theoretical backing for previous empirical findings that PSJF-E performs very well with estimated job sizes.  While SRPT-B provides an additional promising alternative that will sometimes perform better in practice, we recommend PSJF-E as a simple, natural scheduling algorithm that seems to generally perform the best or near the best among standard alternatives.

Our work leaves several open directions, including considering other estimation models, optimizing the behavior of the bounce in SRPT-B, and improving the bounds.
For instance,
another important estimation model is a model where estimates are typically good,
but not guaranteed to be good, such as Gaussian errors.
One might try to adapt our results to that setting
by bounding the worst expected error
over a given interval of time.
We also note it may be possible to tighten the bound on the ratio between PSJF and SRPT, which may correspondingly tighten the bound between PSJF-E and SRPT.

\bibliographystyle{plainurl}
\bibliography{refs}

\appendix

\section{Handling Rank Function Ties}
\label{sec:ties}

One might ask how to handle two jobs that have the same rank under rank-based policies.  For example, it is possible that two jobs of the same size arrive and await service;  when service becomes available, should one be given priority, or should a processor-sharing based policy be used?  We avoid such questions by assuming continuous distributions for size and estimated size, so almost surely
no two jobs arrive with the same values for these quantities.

If one considers a discrete job size distribution,
then PSJF-E in the presence of arbitrarily small errors
is equivalent to PSJF with random priority preemptive tiebreaking,
which can have significantly worse mean response time
than a standard tiebreaking rule such as First-Come-First-Served.
This issue is sidestepped by our focus on continuous distributions.

Alternatively, ties could occur if a job being served while its rank was increasing reached the rank of another job which also would have increasing rank if it was served.
Whichever job is served immediately loses minimum rank status, resulting in a
processor-sharing effect. The tiebreaking policy is therefore irrelevant.
The only scheme analyzed in this paper with increasing rank, SRPT-B, avoids this scenario by capping the rank at the initial size estimate.
Under SRPT-B, a job being served with increasing rank can only be preempted by a newly arriving job, which will finish before the preempted job can continue.

If one removed the assumption of continuous distributions,
we believe that our consistency result for SRPT-B, namely
\cref{thm:srpt-b_result}(b), would still hold,
though the proof would likely be significantly more complicated.
We believe this because the ``arbitrarily small error''
scenario that is damaging to PSJF-E
is irrelevant to SRPT-B
when $\alpha$ and $\beta$ are near 1.
SRPT-B has initially decreasing rank,
so arbitrarily small errors lead to
random \emph{nonpreemptive} tiebreaking,
which does not increase response time.
In contrast, we believe that our results for PSJF-E, namely \cref{thm:psjf-e_result},
would not hold as stated for general (non-continuous) distributions,
due to these complications.
However, we believe that a weaker version of our result could still be proven,
with $1.5$ replaced by a larger constant.

\section{Unachievability of Traditional Robustness}
\label{sec:gittins_bound}
In the context of online algorithms with predictions \cite{DBLP:journals/jacm/LykourisV21,DBLP:books/cu/20/MitzenmacherV20},
one typically tries to achieve constant-factor robustness,
which requires that the approximation ratio is bounded by a constant under arbitrarily poor predictions.

In the context of M/G/1 scheduling with predictions,
constant-factor robustness is provably unachievable.
This follows from the fact that
the mean response time of an algorithm with arbitrarily poor predictions
is bounded below by the optimal blind policy, that has no predictions at all.

In the context of M/G/1 scheduling, the optimal blind policy is known:
the Gittins Index policy \cite{gittins_bandit_1979,scully_gittins_2021}
is optimal in this context,
though it requires knowledge of the job size distribution.
When the job size distribution is unknown,
a policy called RMLF
\cite{kalyanasundaram_minimizing_2003, becchetti_nonclairvoyant_2004}
has a competitive ratio relative to SRPT that grows slowly as a function of load
\cite{bansal_achievable_2018}.

The approximation ratio of the Gittins index policy to the SRPT policy
can be arbitrarily poor in the limit as $\rho \to 1$ \cite{scully_optimal_2020}.
Specifically, whenever the job size distribution $S \in \mathrm{MDA}(\Lambda)$,
the \emph{Gumbel domain of attraction},
the approximation ratio of Gittins to SRPT grows arbitrarily
large as $\rho \to 1$.
Loosely, we can think of $\mathrm{MDA}(\Lambda)$
as including all unbounded job size distributions
whose tail is asymptotically lighter than a power law.
For instance, distributions with exponential, Weibull, and log-normal tails
are in~$\mathrm{MDA}(\Lambda)$.

In all such cases, constant-factor robustness is unachievable.
A better goal might be to compare the performance of a policy
against that of the Gittins policy, or perhaps RMLF.
We leave that question for future work.

\section{Poor Performance of SRPT-B Without the Rank Cap}
\label{sec:SRPT-B-without-cap}

We have discussed that for SRPT-B capping the bounce at initial estimated size of a job is important in our analysis.  We here explain why our results would not hold without such a cap.

Under SRPT-B, if a job begins service,
it can only leave service by being preempted by newly arriving jobs,
not by having its rank rise above the rank of other jobs in the queue.
In contrast, under a policy \emph{SRPT with Unlimited Bounce} (SRPT-UB) with rank function $|z-a|$, without a cap at~$z$,
the situation would be very different.

We consider a case where all jobs have nearly the same size ($S$ is nearly constant) and $\beta = 1/2 - \epsilon$. For specificity, let each job have size $s \in [1, 1+\delta]$,
for $\delta \ll \epsilon$.
Let each job's estimate $z = \beta s$.

Under SRPT-UB, consider a job $j$ of size $1$. It begins service with rank $\beta = 1/2-\epsilon$, descends to rank 0, and rises back to a rank of $1/2-\epsilon$
at age $1-2\epsilon$.
At this age, or in the next $\delta$ service,
job $j$'s rank will rise high enough that it will be preempted by any fresh job (that has yet to receive service),
whether that job arrived before or after job $j$.
As a result, job $j$ will have to wait until there are no more fresh jobs to complete.
Approximately, job $j$ must wait until the system empties to complete.
From standard results in queueing theory,
the mean response time under SRPT-UB is therefore $\Theta\gp[\big]{\frac{1}{(1-\rho)^2}}$.
In contrast, a better policy such as SRPT or SRPT-B
has mean response time that grows as $\Theta\gp[\big]{\frac{1}{1-\rho}}$.
The gap between these two response times grows arbitrarily wide as $\rho \to 1$.

\section{SRPT with Continuously Updating Estimates (SRPT-CUE)}
\label{sec:continuous-estimate}

The SRPT-SE scheme requires knowledge of the job's true size,
and hence while it is useful in our analysis,
it is not a practical policy as we have defined it.
However, there may be settings where a variation of SRPT-SE
is implementable.
We now describe such a setting,
the variant of SRPT-SE that can be implemented in it,
and how to extend our bounds in \cref{thm:srpt-se_result}
to cover the new variant.

Suppose that jobs are given size estimates
not just when they initially arrive
but also continuously during service.
More specifically, suppose that once a job of true size~$s$ reaches each age~$a$,
it is given an estimated remaining size in the interval
$[\beta(s - a), \alpha(s - a)]$.
This could model settings where there is some visible metric of job progress
but the speed at which progress is made is uncertain,
such as sending files of known size to clients with unknown packet loss rates.

In this setting, a natural policy is one we call
\emph{SRPT with Continuously Updating Estimates} (SRPT-CUE),
which at every moment in time serves the job
with the smallest remaining size estimate.
The difference between SRPT-SE and SRPT-CUE is that
while each job's rank decreases linearly under SRPT-SE,
a job's rank may follow a more complicated path under SRPT-CUE,
though it will stay in the interval $[\beta(s - a), \alpha(s - a)]$.
We illustrate the difference in \cref{fig:srpt-cue}.

\begin{figure}
    \begin{subfigure}[t]{0.5\linewidth}
        \begin{tikzpicture}[figure]
    \yguide[estimated size~$z$]{0}{4}
    \xguide[true size~$s$]{7}{0}
    \axes{10}{6}{$0$}{age}{$0$}{rank}
    \draw[srpt-se] (0, 4) -- (7, 0);
\end{tikzpicture}
        \caption{SRPT with Scaling Estimates (SRPT-SE)}
    \end{subfigure}%
    \begin{subfigure}[t]{0.5\linewidth}
        \begin{tikzpicture}[figure]
    \fill[black!14] (0, 2) -- (7, 0) -- (0, 10);
    \yguide[estimated size~$z$]{0}{4}
    \yguide[$\alpha s$]{0}{10}
    \yguide[$\beta s$]{0}{2}
    \xguide[true size~$s$]{7}{0}
    \draw[guide] (0, 2) -- (7, 0);
    \draw[guide] (0, 10) -- (7, 0);
    \axes{10}{10}{$0$}{age}{$0$}{rank}
    \draw[size, rounded corners=0.5em] (0, 4) -- (1.4, 6) -- (2.8, 3) -- (4.2, 3.5) -- (5.6, 0.5) -- (7, 0);
\end{tikzpicture}
        \caption{SRPT with Continuously Updating Estimates (SRPT-CUE)}
    \end{subfigure}
    \caption{Comparison Between SRPT-SE and SRPT-CUE}
    \label{fig:srpt-cue}
\end{figure}
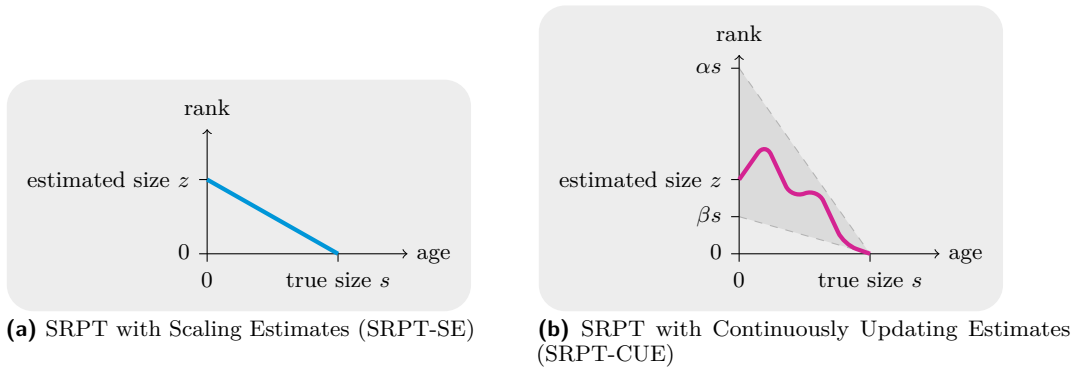

The proof of \cref{thm:srpt-se_result} can be modified to give the same result for SRPT-CUE as for SRPT-SE, namely
\begin{equation}
    \label{eq:srpt-cue}
    \E{T_\srptcue} \leq \frac{\alpha}{\beta} \E{T_\srpt}.
\end{equation}
The two main steps are
(a) formalizing the definition of SRPT-CUE using a rank function and
(b) showing that \cref{thm:soap_match_minimizes_work} holds for SRPT-CUE.
The difficulty of these steps depends on the details
of the estimation error model.
For example, if the estimated remaining size functions for each job
are sampled i.i.d. from some function distribution,
then (a) can be done using methods of \citet{scully_soap_2018},
and (b) follows for the same reasons as the policies we consider.

It may be possible to handle even adversarial estimation errors,
provided they stay $(\beta, \alpha)$-bounded,
by using the methods of \citet{scully_soap_2018-1}.
In this model, (a) is done by assigning each job state
a rank \emph{interval} instead of a single rank,
from which the job's rank may be adversarially chosen.
However, (b) may require placing some limit on the adversary's power,
such as making them oblivious to the system state.

We note that for \cref{eq:srpt-cue} to hold,
it is important that under SRPT-CUE,
a job's rank changes only while that job is in service.
For example, if a job's rank can change while it is not in service,
then small fluctuations in rank might cause the system to split its effort
between two jobs of similar remaining size,
which is worse than serving one of the jobs before the other.
See \citet[proof of Theorem~6]{scully_soap_2018-1} for further discussion.

\end{document}